\documentclass[preprint,authoryear,3p,12pt]{elsarticle}

\usepackage{amssymb}
\usepackage{amsthm}

\usepackage{mathrsfs}
\usepackage{float}
\usepackage{eufrak}

\usepackage{graphics}
\usepackage{amsfonts}
\usepackage{mathrsfs}
\usepackage{amscd}
\usepackage{color}
\usepackage{url}
\usepackage[plainpages=false,pdfpagelabels=true,colorlinks=true,citecolor=blue,hypertexnames=false]{hyperref}

\journal{}

\newtheorem{theorem}{Theorem}[section]
\newtheorem{prop}[theorem]{Proposition}
\newtheorem{cor}[theorem]{Corollary}
\newtheorem{lemma}[theorem]{Lemma}

\newtheorem{define}[theorem]{Definition}
\newtheorem{example}[theorem]{Example}

\floatstyle{ruled}
\newfloat{algorithm}{!h}{loa}
\floatname{algorithm}{Algorithm}
\newcommand{\Input}[1]
  {\noindent\begin{tabular}{@{}p{1.8cm}@{}p{13.2cm}@{}}
   {\bf Input: }&#1 \end{tabular}}
\newcommand{\Output}[1]
  {\noindent\begin{tabular}{@{}p{1.8cm}@{}p{13.2cm}@{}}
   {\bf Output: }&#1 \end{tabular}}

\newfloat{function}{!h}{loa}
\floatname{function}{Function}

\newcommand{\fuluzihao}{\fontsize{9pt}{\baselineskip}\selectfont}

\def\Q{{\mathbb{Q}}}

\def\k{{\rm K}}

\def\lm{{\rm lm}}

\def\lpp{{\rm lpp}}
\def\lc{{\rm lc}}
\def\lcm{{\rm lcm}}

\def\max{{\rm max}}

\def\f{{\bf f}}
\def\e{{\bf e}}

\def\u{{\bf u}}
\def\v{{\bf v}}
\def\w{{\bf w}}
\def\t{{\bf t}}

\def\fu{{f^{[\u]}}}
\def\gv{{g^{[\v]}}}
\def\hw{{h^{[\w]}}}
\def\bru{\bar{\u}}
\def\brv{\bar{\v}}
\def\brf{\bar{f}}
\def\brg{\bar{g}}
\def\brfu{{\brf^{[\bru]}}}
\def\brgv{{\brg^{[\brv]}}}

\def\rp{{\rm Representation}}
\def\psf{{\rm IncompleteStandardForm}}
\def\x{{x_1,\cdots,x_n}}

\def\lla{{\longleftarrow}}

\newcommand{\spc}{\hspace*{15pt}}

\newcommand{\comment}[1]{}
\newcommand{\ignore}[1]{}
\newcommand{\gr}{Gr\"obner\,}

\begin{document}

\begin{frontmatter}



\title{Solving Detachability Problem for the Polynomial Ring by Signature-based Gr\"obner Basis Algorithms}


\author{Yao Sun, Dingkang Wang\fnref{label1}}

\fntext[label1]{The authors are supported by NKBRPC 2011CB302400, NSFC 10971217 and 60821002/F02.}

\address{Key Laboratory of Mathematics Mechanization\\ Academy of Mathematics and Systems Science, Chinese Academy of Sciences, Beijing 100190, China}

\ead{sunyao@amss.ac.cn, dwang@mmrc.iss.ac.cn}

\begin{abstract}

Signature-based algorithms are a popular kind of algorithms for computing \gr basis, including the famous F5 algorithm, F5C, extended F5, G$^2$V and the GVW algorithm. In this paper, an efficient  method is proposed to solve the detachability problem.  The new method only uses the {\em outputs} of signature-based algorithms, and no extra \gr basis computations are needed.  When a \gr basis is obtained by signature-based algorithms, the detachability problem can be settled  in polynomial time.
\end{abstract}


\end{frontmatter}




\section{Introduction}

{\bf The detachability problem of a polynomial ring $R$:} Let $F\subset R$ be a finite subset of polynomials and $f\in R$ be a polynomial. Then
\begin{enumerate}

\item Decide whether $f\in \langle F\rangle$, where $\langle F \rangle$ is the ideal generated by $F$ over $R$.

\item If $f\in \langle F\rangle$, then compute a {\em representation} of $f$ w.r.t. $F$, i.e. compute several polynomials, say $p_i$'s, such that $f=\sum p_if_i$ where $f_i\in F$.

\end{enumerate}
It is well known that solving the detachability problem can be used for solving several other problems, for example, computing a basis for the syzygy module of finite polynomials in $R$, and so on.

The first part of the detachability problem can be solved if a \gr basis for the ideal $\langle F \rangle$ is obtained. Let $G$ be a \gr basis for $\langle F \rangle$, then $f \in \langle F \rangle$ if and only if $f$ is reduced to $0$ by $G$, where the definition of reduction can be found in many books, such as \citep{CLO04}. For the second part of the detachability problem, it suffices to compute a representation for each polynomial $g$ in this \gr basis $G$ w.r.t. $F$. That is, if we have $g_j =\sum q_{j,i}f_i$ for all $g_j\in G$ where $f_i \in F$, then for any $f\in \langle F\rangle$, it is easy to compute $p_j$'s such that $f= \sum p_j g_j$, since $G$ is a \gr basis for $\langle F \rangle$. Therefore, $$f = \sum p_j g_j = \sum_j p_j (\sum_i q_{j,i} f_i) = \sum_i (\sum_j p_j q_{j,i}) f_i$$ is a representation of $f$ w.r.t. $F$.

Therefore, the detachability problem can be settled if {\em representations of polynomials in a \gr basis $G$ w.r.t. $F$} are obtained.

However, if $F$ is not a \gr basis for $\langle F \rangle$, it is very expensive to compute a representation for each $g \in G$ w.r.t. $F$. As we know, a general method for computing representations of polynomials in a \gr basis $G$ w.r.t. $F$ is similar to the method of computing inverse matrices. Specifically, each polynomial $f_i\in F$ is replaced by a new polynomial $f_i + \e_i$ where $\e_i$ is a new variable. Then when computing a \gr basis for the $f_i$ part, the $\e_i$ part records the corresponding track. At last, representations of polynomials in a \gr basis can be obtained directly from these tracks that are recorded by $\e_i$ part. Clearly, this general method needs much more computations than simply computing a \gr basis for $\langle F \rangle$, since more variables ($\e_i$'s) are involved. For more details about this general method, readers are referred to one of the books \citep{Mishra93, Greuel02}.

In current paper, a new method is presented to compute representations of polynomials in a \gr basis w.r.t. the ideal generators. This new method mainly uses the properties of the {\em outputs} of signature-based algorithms, and does not need to revise these signature-based algorithms. Recently, signature-based algorithms are a popular kind of algorithms for computing \gr basis, including the famous F5 algorithm \citep{Fau02}, F5C \citep{Eder09}, extended F5 \citep{Ars09}, G$^2$V \citep{Gao09} and the GVW algorithm \citep{Gao10b}. Related researches on signature-based algorithms are also found in \citep{Stegers05, Eder08, Albrecht10, Arri10, SunWang10a, SunWang10b, SunWang11, SunWang11b, Zobnin10, Huang10, Eder11}.

The F5 and GVW algorithms have already been used to compute a basis for the syzygy module of $F$. In \citep{Ars11}, the F5 algorithm is revised to keep track of multiples of polynomials used in the reduction, and then a basis for the syzygy module of $F$ can be obtained from these tracks. However, this method slows down the F5 algorithm for computing \gr basis for $\langle F \rangle$, since keeping tracks is very costly in both memory and execution time. Gao et al. give a method to get representations of polynomials from the output of the GVW algorithm in \citep{Gao10b}. By their method, a new set of polynomials is obtained, and representations of this new set of polynomials w.r.t. $F$ are obtained at the same time. However, this new set of polynomials should be a \gr basis for $\langle F \rangle$. According to the detailed procedure of GVW, the new set of polynomials, which are constructed by their method from the output of GVW, can be proved to be a \gr basis for $\langle F \rangle$ easily. But it is very difficult to prove that the new set of polynomials, which are constructed by their method from the output of F5, is also a \gr basis for $\langle F \rangle$. So their method cannot be applied to F5 directly.



In this paper, the new method to compute representations of polynomials in a Gr\"obner basis w.r.t. $F$ is based on a new notion for the ideal $\langle F \rangle$: {\em labeled \gr basis}, from which representations of polynomials in a \gr basis are obtained directly. It is very expensive to compute a labeled \gr basis since it contains too much information.  A labeled \gr basis is mainly used in theoretical analysis and is usually not returned by a signature-based algorithm. For sake of efficiency, all existing signature-based algorithms (including F5 and GVW) output simpler versions of labeled \gr basis. So the main work of this paper is to construct labeled \gr bases from these simpler versions.

When the representations of polynomials in a \gr basis $G$ w.r.t. $F$ are obtained, one can also compute a basis for the syzygy module of $F$ easily. There are two efficient approaches. First, if the signatures of polynomials that are reduced to $0$ have been kept in the signature-based algorithms. A basis for the syzygy module of $F$ can be {\em recovered} in polynomial time by using a similar method mentioned in \citep{Gao10b}. Second, since the matrices $A$ and $B$ such that $G = FA$ and $F = GB$ can be got from these representations directly, Buchberger has presented an efficient algorithm to get a basis for the syzygy module of $F$ from a basis for the syzygy module of $G$ in \citep{Buchberger85}. Buchberger's algorithm can also be found in \citep{Mishra93, Greuel02, CLO04}. Note that Buchberger's algorithm is also a polynomial time algorithm if the matrices $A$ and $B$ are known.

This paper is organized as follows. Some preliminaries are presented in Section \ref{sec_pre}. Section \ref{sec_gbrep} shows in detail how to obtain representations of polynomials in a \gr basis from the outputs of signature-based algorithms. An example is given in Section \ref{sec_exa} to show how our method works and concluding remarks follow in Section \ref{sec_con}. Some related proofs are put in \ref{sec_f5}.

\section{Preliminaries} \label{sec_pre}

\subsection{Notations}

Let $R:=\k[\x]$ be a polynomial ring over a field $\k$ with $n$ variables. Suppose $F := \{f_1, \cdots, f_m\}$ is a finite subset of $R$ and $$I:=\langle f_1, \cdots, f_m\rangle=\{p_1f_1+\cdots+p_mf_m \mid p_1,\cdots,p_m \in R\}$$ is the ideal generated by $F$.


Fix a term order $\prec_1$ on $R$. We denote the {\em leading power product}, {\em leading coefficient} and {\em leading monomial} of a polynomial $f\in R$ by $\lpp(f)$, $\lc(f)$ and $\lm(f)$ respectively. For example, let $f := 2x^2y+3z \in \Q[x, y, z]$ be a polynomial and $\prec_1$ be the Graded Reverse Lex order with $z \prec_1 y \prec_1 x$, where $\Q$ is the rational number field. Then $\lpp(f) = x^2y$, $\lc(f) = 2$ and $\lm(f) = 2x^2y$. Note that   $\lm(f) = \lc(f) \lpp(f)$ always holds.


In signature-based algorithms, the map $R^m \longrightarrow I$: $$(p_1, \cdots, p_m) \longmapsto p_1f_1+\cdots+p_mf_m, $$ has been extensively used. Let $\f := (f_1, \cdots, f_m)\in R^m$. Then for any $f\in I$, there always exists (at least) a vector $\u=(p_1, \cdots, p_m)\in R^m$ such that $$f = \u \cdot \f = p_1f_1+\cdots+p_mf_m, $$ where ``$\cdot$" is the inner product of two vectors. Note that such vector $\u$ is not unique. For example, let $(f_1, f_2, f_3) := (yz-x, xz-y, xy-z) \subset \Q[x, y, z]^3$ and $f := y^2-z^2 \in \langle f_1, f_2, f_3 \rangle$, then $f=(0, -y, z) \cdot (f_1, f_2, f_3) = (xz-y, -yz+x-y, z) \cdot (f_1, f_2, f_3)$.

Given $f\in I$ and $\u\in R^m$ such that $f=\u\cdot \f$, we use the notation $\fu$ to express this relation between $f$ and $\u$. Computations on $\fu$ can be defined naturally. Let $\fu$ and $\gv$ be such that $f=\u \cdot \f$ and $g = \v \cdot \f$, $c$ be a constant in $K$ and $t$ be a power product in $R$. Then
\begin{enumerate}
\item $\fu + \gv = (f+g)^{[\u+\v]}$.

\item $ct (f^{[\u]}) = (ctf)^{[ct\u]}$.
\end{enumerate}
The above operations are well defined, since $f+g = (\u+\v) \cdot \f$ and $ctf = (ct\u) \cdot \f$. In fact, the above $\fu$ and $\gv$ are both elements of the following $R$-module: $$\{\fu \mid f=\u \cdot \f \mbox{ and } \u\in R^m\} \mbox{ or equivalently } \{p_1f_1^{[\e_1]} + \cdots + p_m f_m^{[\e_m]} \mid p_1, \cdots, p_m \in R\},$$ where $\e_i$ is the $i$-th unit vector of $R^m$, i.e. $(\e_i)_j=\delta_{ij}$ where $\delta_{ij}$ is the Kronecker delta.

To make the notation $\fu$ easier to understand, {\bf we also call $\fu$ to be a polynomial in $I$ and write $\fu\in I$.} Besides, {\bf the notation $\fu$ always means $f\in I$ and $f=\u\cdot \f$ in this paper}. For two polynomials $\fu$ and $\gv$ in $I$, we say $\fu=\gv$ only when $f = g$ and $\u = \v$.

Fix {\em any} term order $\prec_2$ on $R^m$. We must emphasize that the order $\prec_2$ may or may not be related to $\prec_1$ in theory, although $\prec_2$ is usually an extension of $\prec_1$ to $R^m$ in implementation. We define the {\em leading power product}, {\em leading coefficient} and
{\em leading monomial} of $\u \in R^m$ w.r.t. $\prec_2$ to be $\lpp(\u)$, $\lc(\u)$ and $\lm(\u)$. More related terminologies on ``module" can be found in Chapter 5 of \citep{CLO04}.

For sake of convenience, we use $\prec$ to represent $\prec_1$ and $\prec_2$, if no confusion occurs. We make the convention that if $f=0$ then $\lpp(f)=0$ and $0 \prec t$ for any non-zero power product $t$ in $R$; similarly for $\lpp(\u)$.

For any polynomial $\fu \in I$ where $f=\u \cdot \f$, we define $\lpp(\u)$ to be the {\bf signature} of $\fu$. Original definition of signature is introduced by Faug\`ere in \citep{Fau02}, and recently, Gao et al. give a generalized definition of signature in \citep{Gao10b}. The above definition is given by Gao et al.

\subsection{Full-labeled \gr basis, monomial-labeled Gr\"obner basis and signature-labeled Gr\"obner basis}

Let $$G := \{g_1^{[\v_1]}, \cdots, g_s^{[\v_s]}\}$$ be a finite subset of $I$. We call $G$ a {\bf labeled Gr\"obner basis} or {\bf full-labeled Gr\"obner basis} for $I$, if for any $\fu \in I$ with $f\not=0$, there exists $\gv\in G$ such that
\begin{enumerate}

\item $\lpp(g)$ divides $\lpp(f)$, and

\item $\lpp(t\v) \preceq \lpp(\u)$, where $t=\lpp(f)/\lpp(g)$.
\end{enumerate}

\begin{prop} \label{prop_gb}
If $G$ is a full-labeled \gr basis for $I$, then the set $\{g \mid \gv \in G\}$ is a \gr basis of the ideal $I=\langle f_1, \cdots, f_m\rangle$.
\end{prop}

\begin{proof}
For any $f\in \langle f_1, \cdots, f_m\rangle$, there exist $p_1, \cdots, p_m \in R$ such that $f=p_1f_1 + \cdots + p_m f_m$. Let $\u:=(p_1, \cdots, p_m)$. Then $\fu \in I$ and hence there exists $\gv\in G$ such that $\lpp(g)$ divides $\lpp(f)$ by the definition of full-labeled \gr basis.
\end{proof}

However, the reverse of the above proposition is usually not true.

\begin{example} \label{exa_lgb}
Let $\f := (f_1, f_2, f_3) = (xz-y, y^2+xz, 2xy+2x)\in \Q[x,y,z]^3$ where $\Q$ is the rational field, and $I$ be the ideal generated by $\{f_1, f_2, f_3\}$. The order $\prec_1$ on $\Q[x,y,z]$ is the Graded Reverse Lex order with $x \succ_1 y \succ_1 z$, and the order $\prec_2$ on $\Q[x,y,z]^3$ is extended from $\prec_1$ in a position over term fashion, i.e.
$$x^\alpha\e_i \prec_2 x^\beta\e_j    \mbox{  iff  } \left\{\begin{array}{l} i > j, \\ \mbox{ or } \\ i = j \mbox{ and } x^\alpha \prec_1 x^\beta. \end{array}\right.$$
\end{example}

It is evident that $\{f_1, f_2, f_3\}$ itself is a \gr basis for the ideal $\langle f_1, f_2, f_3\rangle$. But the set $G=\{f_1^{[\e_1]}, f_2^{[\e_2]}, f_3^{[\e_3]}\}$ is {\em not} a full-labeled \gr basis for $I$. The reason is that, there exists a polynomial $(2x^2z-2xy)^{[2x\e_2-y\e_3]} = 2x(f_2^{[\e_2]}) - y(f_3^{[\e_3]}) \in I$, and $f_1^{[\e_1]}$ is the only polynomial in $G$ such that $\lpp(f_1) = xz$ divides $\lpp(2x^2z-2xy)=x^2z$. But $x\e_1 \succ_2 \lpp(2x\e_2-y\e_3) = x\e_2$. The readers can check that the set $\{f_1^{[\e_1]}, f_2^{[\e_2]}, f_3^{[\e_3]}, (2x^2z-2xy)^{[2x\e_2-y\e_3]}\}$ is a full-labeled \gr basis for $I$.

However, for sake of efficiency, all signature-based algorithms, including F5 and GVW, do not return a full-labeled \gr basis. Next, we introduce two derived conceptions.

\begin{enumerate}

\item Let $M := \{g_1^{\{c_1\t_1\}}, \cdots, g_s^{\{c_s\t_s\}}\}$ be a finite set, where $c_i$ is a constant in $K$, $\t_i$ has the form of $x^{\alpha_i} \e_j$ and $g_i$ is a polynomial in $R$. The set $M$ is called a {\bf monomial-labeled \gr basis} for $I$, if there exist $\v_1, \cdots, \v_s \in R^m$ such that $g_i = \v_i \cdot (f_1, \cdots, f_m)$, $\lm(\v_i) = c_i \t_i$ and $\{g_1^{[\v_1]}, \cdots, g_s^{[\v_s]}\}$ is a full-labeled \gr basis for $I$.

\item Similarly, let $S := \{g_1^{(\t_1)}, \cdots, g_s^{(\t_s)}\}$ be a finite set, where $\t_i$ has the form of $x^{\alpha_i} \e_j$ and $g_i$ is a polynomial in $R$. The set $S$ is called a {\bf signature-labeled \gr basis} for $I$, if there exist $\v_1, \cdots, \v_s \in R^m$ such that $g_i = \v_i \cdot (f_1, \cdots, f_m)$, $\lpp(\v_i) = \t_i$ and $\{g_1^{[\v_1]}, \cdots, g_s^{[\v_s]}\}$ is a full-labeled \gr basis for $I$.

\end{enumerate}
Note that the notation $g^{[\cdot]}$ is used in full-labeled \gr basis, notation $g^{\{\cdot\}}$ is used in monomial-labeled \gr basis, and $g^{(\cdot)}$ is used in signature-labeled \gr basis.

The only difference between the above two derived conceptions is that, the coefficients of $\v_i$'s are kept in a monomial-labeled \gr basis but not kept in a signature-labeled \gr basis. Both of them are simpler version of full-labeled \gr basis. For instance, in Example \ref{exa_lgb}, the set $\{f_1^{\{\e_1\}}, f_2^{\{\e_2\}}, f_3^{\{\e_3\}}, (2x^2z-2xy)^{\{2x\e_2\}}\}$ is a monomial-labeled \gr basis for $I$ and $\{f_1^{(\e_1)}, f_2^{(\e_2)}, f_3^{(\e_3)}, (2x^2z-2xy)^{(x\e_2)}\}$ is a signature-labeled \gr basis for $I$.

In practical implementation, the G$^2$V and GVW algorithms return monomial-labeled \gr bases, which is shown in \citep{Gao10b}. We will prove in \ref{sec_f5} that F5 computes a signature-labeled \gr basis, and the proofs can also be applied to the variants of F5 after minor revisions. So almost all existing signature-based algorithms compute monomial-labeled \gr bases or signature-labeled \gr bases.




\section{Computing Representations of Polynomials in a Gr\"obner Basis} \label{sec_gbrep}

Let $F := \{f_1, \cdots, f_m\} \subset R$ and $I$ be the ideal generated by $F$. This section is organized as follows. Subsection \ref{subsec_lgb} shows how to express the polynomials in a \gr basis for $I$ as the linear combinations of the polynomials in $F$ with coefficient in $R$ from a full-labeled \gr basis; Subsection \ref{subsec_mlgb} details how to build a full-labeled \gr basis from a monomial-labeled \gr basis; Subsection \ref{subsec_slgb} describes how to construct a monomial-labeled \gr basis from a signature-labeled \gr basis.

\subsection{Express polynomials in a \gr basis as the linear combinations of ideal generators from a full-labeled \gr basis}\label{subsec_lgb}

Let $G := \{g_1^{[\v_1]}, \cdots, g_s^{[\v_s]}\}$ be a full-labeled \gr basis for $I$. Then the set $G_0 = \{g_1, \cdots, g_s\}$ is a \gr basis for $I$ by Proposition \ref{prop_gb}. Moreover, the equation $g_i = \v_i \cdot \f$ always holds by definition where $\f=(f_1, \cdots, f_m)$. Thus, $\v_i$ provides a representation of $g_i$ w.r.t. $F$ for each $g_i\in G_0$.

Regarding the detachability problem, suppose $f$ is a polynomial in $R$. If $f$ is not reduced to $0$ by $G_0$, then $f\notin I$; otherwise, $f\in I$ and there exist $p_1, \cdots, p_s \in R$, such that $$f = p_1g_1 + \cdots + p_s g_s,$$ where $g_i \in G_0$ for $i = 1, \cdots, s$. Let $\u := p_1\v_1 + \cdots + p_s \v_s$. Then we have
$$\u\cdot \f = (p_1\v_1 + \cdots + p_s \v_s) \cdot \f = p_1\v_1 \cdot \f + \cdots + p_s \v_s \cdot \f = p_1g_1 + \cdots + p_s g_s = f.$$
The vector $\u$ provides a representation of $f$ w.r.t. $\{f_1, \cdots, f_m\}$.

Particularly, inter-reducing polynomials in $G_0$ can generate the reduced \gr basis for $I$. Since all polynomials in the reduced \gr basis are elements in $I$, representations of polynomials in the reduced \gr basis for $I$ w.r.t. $F$ can be obtained similarly.


\subsection{Build a full-labeled \gr basis from a monomial-labeled \gr basis} \label{subsec_mlgb}

Let $M := \{g_1^{\{c_1\t_1\}}, \cdots, g_s^{\{c_s\t_s\}}\}$ be a monomial-labeled \gr basis for $I$. Then by definition there exists a full-labeled \gr basis $G = \{g_1^{[\v_1]}, \cdots, g_s^{[\v_s]}\}$ for $I$ such that $\lm(\v_i) = c_i \t_i$. In this subsection, we show how to build a full-labeled \gr basis from $M$, i.e. to compute the polynomials $\v_1,\cdots,\v_s$ such that $\lm(\v_i) = c_i \t_i$. The following proposition plays an important role in this procedure.

\begin{prop} \label{prop_mlgb2lgb}
Let $G := \{g_1^{[\v_1]}, \cdots, g_s^{[\v_s]}\}$ be a full-labeled \gr basis for $I$ which is generated by $\{f_1, \cdots, f_m\}$. If $cx^\alpha \e_j$ is a monomial and $f$ is a polynomial in $I$ such that there exists $\u \in R^m$ with $\lm(\u) = cx^\alpha \e_j$ and $f=\u \cdot \f$ where $\f=(f_1, \cdots, f_m)$, then there exist polynomials $p_1, \cdots, p_s \in R$ such that $$f = cx^\alpha f_j + p_1 g_1 + \cdots + p_s g_s,$$ and $x^\alpha \e_j \succ \lpp(p_i\v_i)$ for $i = 1, \cdots, s$.

Moreover, with the above $p_i$'s, let $$\u' := cx^\alpha\e_j + p_1 \v_1 + \cdots + p_s \v_s.$$ Then we have $\lm(\u') = cx^\alpha\e_j$ and $f = \u' \cdot \f$.
\end{prop}

\begin{proof}
We present a constructive method for finding the desired $p_i$'s. Initially, all $p_i$'s are set to be $0$.

Consider the polynomial $\hw := \fu - (cx^\alpha f_j)^{[cx^\alpha \e_j]} \in I$. Note that $\w$ is unknown,  since $\u$ is unknown, but such $\w$ does exist and we know $\lpp(\w) \prec \lpp(\u) = x^\alpha \e_j$. In the following, we will {\em not} use the value of $\w$ and we only use the properties that $h = \w \cdot \f$ and $\lpp(\w) \prec x^\alpha \e_j$ where $x^\alpha \e_j$ is known.

We next reduce the polynomial $h$ to $0$ with polynomials in $G$. Specifically, if $h=0$, then $\{p_i = 0 \mid i = 1, \cdots, s\}$ are the desired polynomials; otherwise, since $\hw \in I$, then according to the definition of a full-labeled \gr basis, there exists some $g_i^{[\v_i]} \in G$ such that $\lpp(g_i)$ divides $\lpp(h)$ and $\lpp(t\v_i) \preceq \lpp(\w) \prec x^\alpha \e_j$ where $t=\lpp(h)/\lpp(g_i)$. Denote $h_1^{[\w_1]} := \hw - (\lm(h)/\lm(g_i)) (g_i^{[\v_i]}) \in I$. Clearly, we still have $h_1 = \w_1 \cdot \f$ and $\lpp(\w_1) \preceq \lpp(\w) \prec x^\alpha \e_j$. In order to obtain the desired $p_i$'s at last, we now update the value of $p_i$ by $p_i + \lm(h)/\lm(g_i)$. If $h_1\not=0$, then we repeat the above process. This process must terminate after finite steps, since the term order on $R$ is a well order. Suppose $h_l^{[\w_l]}$ is the last polynomial, then $h_l = 0$ must hold.

With the $p_i$'s obtained in above procedure, we have $f - cx^\alpha f_j = p_1 g_1 + \cdots + p_s g_s$ and $x^\alpha \e_j \succ \lpp(p_i\v_i)$ for $i = 1, \cdots, s$. This   proves the first part of the proposition.

For the second part of the proposition, let $p_i$'s be the polynomials obtained above and let $\u' := cx^\alpha\e_j + p_1 \v_1 + \cdots + p_s \v_s$. We have $\lm(\u') = cx^\alpha\e_j$, since $x^\alpha \e_j \succ \lpp(p_i\v_i)$. And $\u' \cdot \f = (cx^\alpha\e_j + p_1 \v_1 + \cdots + p_s \v_s) \cdot \f = cx^\alpha\e_j \cdot \f + p_1 \v_1\cdot \f + \cdots + p_s \v_s \cdot \f = cx^\alpha f_j + p_1 g_1 + \cdots + p_s g_s = f$.
\end{proof}


From the proof of the above proposition, we find that not all of the polynomials in a full-labeled \gr basis $G$ are necessary during the procedure of constructing these $p_i$'s. So we have the following direct consequence.

\begin{cor}\label{cor_mlgb2lgb}
Let $G := \{g_1^{[\v_1]}, \cdots, g_s^{[\v_s]}\}$ be a full-labeled \gr basis for $I$ which is generated by $\{f_1, \cdots, f_m\}$. Let $cx^\alpha \e_j$ be a monomial, $f$ be a polynomial in $I$ such that there exists $\u \in R^m$ with $\lm(\u) = cx^\alpha \e_j$ and $f=\u \cdot \f$ where $\f=(f_1, \cdots, f_m)$, and $G_{\prec x^\alpha \e_j}$ be the set $\{g_i^{[\v_i]} \in G \mid \lpp(\v_i) \prec x^\alpha \e_j\}\subset G$.

Then there exist polynomials $p_1, \cdots, p_l \in R$ and ${g'_1}^{[\v'_1]}, \cdots, {g'_l}^{[\v'_l]} \in G_{\prec x^\alpha \e_j}$ such that $f = cx^\alpha f_j + p_1 g'_1 + \cdots + p_l g'_l$ and $x^\alpha \e_j \succ \lpp(p_i\v'_i)$ for $i = 1, \cdots, l$. Moreover, with the above $p_i$'s, let $\u' := cx^\alpha\e_j + p_1 \v'_1 + \cdots + p_s \v'_s.$ Then we have $\lm(\u') = cx^\alpha\e_j$ and $f = \u' \cdot \f$.
\end{cor}

Based on the above corollary, the following algorithm builds a full-labeled \gr basis from a monomial-labeled \gr basis.

\smallskip
\noindent {\bf Algorithm --- Mono2FullLGB}\\
\Input{$M = \{g_1^{\{c_1\t_1\}}, \cdots, g_s^{\{c_s\t_s\}}\}$, a monomial-labeled Gr\"obner basis for $I$.}\\
\Output{$G = \{g_1^{[\v_1]}, \cdots, g_s^{[\v_s]}\}$, a full-labeled \gr basis for $I$ such that $\lm(\v_i) = c_i \t_i$.}

\begin{enumerate}

\item Let $G := \emptyset$.

\item Choose $f^{\{c x^\alpha \e_j\}}$ from $M$ with $x^\alpha \e_j \preceq \t_i$ for all $g_i^{\{c_i\t_i\}} \in M$.

\item Remove $f^{\{c x^\alpha \e_j\}}$ from $M$, i.e. $M := M \setminus \{f^{\{c x^\alpha \e_j\}}\}$.

\item Compute $p_i$'s by using Function $\rp(c x^\alpha \e_j, f, G)$, such that $f = c x^\alpha f_j + \sum p_i g_i$ where $g_i^{[\v_i]} \in G$ and $x^\alpha \e_j \succ \lpp(p_i\v_i)$.

\item With the above $p_i$'s, let $\u := c x^\alpha \e_j + \sum p_i \v_i$ where $g_i^{[\v_i]} \in G$.

\item Let $G := G \cup \{\fu\}$.

\item If $M$ is empty, then return $G$; otherwise, goto step 2.

\end{enumerate}

Function {\em Representation} is based on the proof of Proposition \ref{prop_mlgb2lgb}.

\smallskip
\noindent {\bf Function --- $\rp(c x^\alpha \e_j, f, G)$}\\
\Input{$c x^\alpha \e_j$, a monomial; $f$, a polynomial in $I$ such that there exists $\u \in R^m$ with $\lm(\u) = c x^\alpha \e_j$ and $f = \u \cdot \f$; $G = \{g_1^{[v_1]}, \cdots, g_t^{[\v_t]}\}$, a subset of $I$.}\\
\Output{$\{p_1, \cdots, p_t\}$, a set of polynomials in $R$ such that $f = c x^\alpha f_j + p_1 g_1 + \cdots + p_t g_t$ and $x^\alpha \e_j \succ \lpp(p_i\v_i)$.}

\begin{enumerate}

\item Let $(p_1, \cdots, p_t) := (0, \cdots, 0)$ and $h := f - c x^\alpha f_j$.

\item If there exists $g_i^{[\v_i]} \in G$ such that $\lpp(g_i)$ divides $\lpp(h)$ and $(\lpp(h)/\lpp(g_i))\lpp(\v_i) \prec x^\alpha \e_j$, then $h := h-(\lm(h)/\lm(g_i))g_i$ and $p_i := p_i + (\lm(h)/\lm(g_i))$.

\item If $h=0$ then return $\{p_1, \cdots, p_t\}$; otherwise, goto step 2.

\end{enumerate}

As discussed in the proof of Proposition \ref{prop_mlgb2lgb}, Function {\em Representation} always terminates in Algorithm Mono2FullLGB.




\subsection{Construct a monomial-labeled \gr basis from a signature-labeled \gr basis} \label{subsec_slgb}

Let $S := \{g_1^{(\t_1)}, \cdots, g_s^{(\t_s)}\}$ be a signature-labeled Gr\"obner basis for $I$. The goal of this subsection is to find coefficients $c_1, \cdots, c_s\in K$ such that $\{g_1^{\{c_1\t_1\}}, \cdots, g_s^{\{c_s\t_s\}}\}$ is a monomial-labeled \gr basis for $I$. For this purpose, we first study an invariant in the ideal $I$. Note that in this section, the notation $g^{(\t)}$ always means there exists $\gv \in I$ such that $\lpp(\v) = \t$.

Given a term $x^\alpha\e_j \in R^m$, we say a polynomial $\gv \in I$ is a {\bf standard form} of $x^\alpha\e_j$, if
\begin{enumerate}
\item[(1)] $\lpp(\v) = x^\alpha\e_j$, and

\item[(2)] $\lpp(g) \preceq \lpp(f)$ for any $\fu \in I$ with $\lpp(\u) = x^\alpha\e_j$.
\end{enumerate}
Note that the polynomial $g$ in the standard form $\gv$ can be zero polynomial. Standard forms of $x^\alpha\e_j$ is not unique in $I$, but for any two standard forms of $x^\alpha\e_j$, we have the following important property.

\begin{prop} \label{prop_standardform}
Let $x^\alpha\e_j$ be a term in $R^m$. If $\gv, g'^{[\v']} \in I$ are two standard forms of $x^\alpha\e_j$. Then $\lm(g)/\lc(\v) = \lm(g')/\lc(\v')$.
\end{prop}

\begin{proof}
By the definition of standard form, we have  $\lpp(g) \preceq \lpp(g')$ and $\lpp(g') \preceq \lpp(g)$ since both $\gv $ and $g'^{[\v']} $ are standard forms of $x^\alpha\e_j$. This follows that  $\lpp(g) = \lpp(g')$.

If $g=g'=0$, the equation holds clearly.  We assume that $g$ and $g'$ are  nonzero in the rest of the proof.

It remains to show $\lc(g)/\lc(\v) = \lc(g')/\lc(\v')$. We will prove this by contradiction.
Assume that $\lc(g)/\lc(\v) \not= \lc(g')/\lc(\v')$. Let $\hw := (1/\lc(g))(\gv) - (1/\lc(g'))(g'^{[\v']}) \in I$. We have $\lpp(h) \prec \lpp(g)$, since $h = (1/\lc(g))g - (1/\lc(g'))g'$ and $\lpp(g) = \lpp(g')$. Since $\w=(1/\lc(g))\v - (1/\lc(g'))\v'$ and $\lc(g)/\lc(\v) \not= \lc(g')/\lc(\v')$, we then have $\lpp(\w) = \lpp(\v) = \lpp(\v') = x^\alpha\e_j$. This means $\hw$ is a polynomial in $I$ such that $\lpp(\w) = x^\alpha\e_j$ and $\lpp(h) \prec \lpp(g) = \lpp(g')$. This contradicts  the fact that both $\gv $ and $g'^{[\v']}$ are standard forms of $x^\alpha\e_j$ and complete the proof of the proposition.
\end{proof}

The above proposition shows that {\em for any standard form $\gv \in I$ of $x^\alpha\e_j$, the monomial $\lm(g)/\lc(\v)$ is an invariant to $x^\alpha\e_j$}.

%
%

Given a term $x^\alpha\e_j$, even if only a signature-labeled \gr basis is known, the standard forms of $x^\alpha\e_j$ can be checked. Clearly, if $\gv \in I$, $\lpp(\v) = x^\alpha\e_j$ and $g = 0$, then $\gv$ is a standard form of $x^\alpha\e_j$ by definition.

\begin{prop}\label{prop_computestandardform}
Let $S := \{g_1^{(\t_1)}, \cdots, g_s^{(\t_s)}\}$ be a signature-labeled Gr\"obner basis for $I$, and $x^\alpha\e_j$ be a term in $R^m$. A polynomial $\gv \in I$ with $\lpp(\v) = x^\alpha\e_j$ and $g\not = 0$, is a standard form of $x^\alpha\e_j$, if and only if there is no $g_i^{(\t_i)}\in S$ such that $\lpp(g_i)$ divides $\lpp(g)$ and $t\t_i \prec x^\alpha\e_j$ where $t=\lpp(g)/\lpp(g_i)$.
\end{prop}

\begin{proof}
Let $G := \{g_1^{[\v_1]}, \cdots, g_s^{[\v_s]}\}$ be a full-labeled \gr basis for $I$ such that $\lm(\v_i) = c_i \t_i$.

On one hand, let $\gv \in I$, where $g\not=0$, be a standard form of $x^\alpha\e_j$. Assume there is some $g_i^{(\t_i)} \in S$ such that $\lpp(g_i)$ divides $\lpp(g)$ and $t\t_i \prec x^\alpha\e_j$ where $t=\lpp(g)/\lpp(g_i)$. Regarding $g_i^{(\t_i)}$, the polynomial $g_i^{[\v_i]}$ is in $I$ and $t\lpp(\v_i) = t\t_i \prec x^\alpha\e_j$ where $t=\lpp(g)/\lpp(g_i)$. Denote $\hw := \gv  - (\lm(g)/\lm(g_i)) (g_i^{[\v_i]})\in I$. We then have $\lpp(h) \prec \lpp(g)$ and $\lpp(\w) = \lpp(\v) = x^\alpha\e_j$, which contradicts that $\gv $ is a standard form of $x^\alpha\e_j$. So there is no $g_i^{(\t_i)} \in S$ such that $\lpp(g_i)$ divides $\lpp(g)$ and $t\t_i \prec x^\alpha\e_j$ where $t=\lpp(g)/\lpp(g_i)$.

On the other hand, let $\gv$ be a polynomial in $I$ with $\lpp(\v) = x^\alpha\e_j$ and $g\not=0$.
Assume $\gv$ is not a standard form of $x^\alpha\e_j$, then by definition, there exists $\fu \in I$ such that $\lpp(\u) = x^\alpha\e_j$ and $\lpp(g) \succ \lpp(f)$. Next, denote $\hw := \gv - (\lc(\v)/\lc(\u))(\fu) \in I$, then we have $\lpp(h) = \lpp(g)$ and $\lpp(\w) \prec \lpp(\v) = x^\alpha\e_j$. Since $\hw$ is a polynomial in $I$, according to the definition of full-labeled \gr basis, there exists $g_i^{[\v_i]} \in G$ such that $\lpp(g_i)$ divides $\lpp(h) = \lpp(g)$ and $\lpp(t\v_i) \preceq \lpp(\w) \prec  x^\alpha\e_j$ where $t=\lpp(g)/\lpp(g_i)$. Note that $\t_i= \lpp(\v_i) $, then $g_i^{(\t_i)}$ is in $S$ such that $\lpp(g_i)$ divides $\lpp(g)$ and $t\t_i = \lpp(t\v_i) \preceq \lpp(\w) \prec x^\alpha\e_j$ where $t=\lpp(g)/\lpp(g_i)$. This is a contradiction. So $\gv$ must be a standard form of $x^\alpha\e_j$.
\end{proof}




With the above proposition, given a signature-labeled \gr basis $S$ for $I$ and a polynomial $\fu \in I$ with $\lpp(\u) = x^\alpha\e_j$, we can check whether $\fu$ is a standard form of $x^\alpha\e_j$.
 The following function, which is deduced from the above proposition, would compute an {\em incomplete} standard form of $x^\alpha \e_j$ from $f^{(x^\alpha \e_j)}$, where $f^{(x^\alpha \e_j)}$ means there exists $\fu \in I$ such that $\lpp(\u) = x^\alpha \e_j$. This incomplete version of standard form is very useful for constructing a monomial-labeled \gr basis from a signature-labeled \gr basis.

\smallskip
\noindent {\bf Function --- $\psf(f^{(x^\alpha \e_j)}, S)$}\\
\Input{$f^{(x^\alpha \e_j)}$, there exists $\fu \in I$ such that $\lpp(\u) = x^\alpha \e_j$; $S = \{g_1^{(\t_1)}, \cdots, g_s^{(\t_s)}\}$, a signature-labeled \gr basis for $I$.}\\
\Output{$g^{(x^\alpha \e_j)}$, there exists $\gv \in I$ such that $\lpp(\v) = x^\alpha \e_j$ and $\gv$ is a standard form of $x^\alpha\e_j$.}

\begin{enumerate}

\item Let $g := f$.

\item If there exists $g_i^{(\t_i)}\in S$ such that $\lpp(g_i)$ divides $\lpp(g)$ and $(\lpp(g)/\lpp(g_i))\t_i \prec x^\alpha \e_j$, then $g:=g-(\lm(h)/\lm(g_i))g_i$; otherwise, return $g^{(x^\alpha \e_j)}$.

\item If $g=0$ then return $0^{(x^\alpha \e_j)}$; otherwise, goto step 2.

\end{enumerate}

Using this incomplete version of standard form, we can now construct a monomial-labeled \gr basis from a signature-labeled \gr basis.




\begin{theorem} \label{thm_sbgb2mbgb}
Let $S  := \{g_1^{(\t_1)}, \cdots, g_s^{(\t_s)}\}$ be a signature-labeled Gr\"obner basis for $I$, and $f^{(x^\alpha\e_j)}$ be such that there exists $\fu \in I$ with $\lpp(\u) = x^\alpha\e_j$. Suppose $g^{(x^\alpha\e_j)}$ is an incomplete standard form of $x^\alpha\e_j$ computed from $f^{(x^\alpha\e_j)}$, and $g_0^{(x^\alpha\e_j)}$ is an incomplete standard form of $x^\alpha\e_j$ computed from $(x^\alpha f_j)^{(x^\alpha\e_j)}$. If $g\not=0$, let $c:=\lc(g)/\lc(g_0)$; otherwise, let $c:=1$. Then there exists $f^{[\u']} \in I$ such that $\lm(\u') = c x^\alpha\e_j$.
\end{theorem}

\begin{proof}
We begin with the case $g=0$, and $c=1$ in this case. By Proposition \ref{prop_standardform}, there must exist $0^{[\v]} \in I$ such that $\lpp(\v) = x^\alpha\e_j$ and $0^{[\v]}$ is a standard form of $x^\alpha\e_j$. As there exists $\fu \in I$ with $\lpp(\u) = x^\alpha\e_j$, let $\hw := \fu  - ((\lc(\u) - 1)/\lc(\v)) (0^{[\v]}) \in I$. Note that $\w = \u - ((\lc(\u) - 1)/\lc(\v))\v = \u - (\lc(\u)/\lc(\v))\v + (1/\lc(\v))\v$ and $\lpp(\u) = \lpp(\v) = x^\alpha\e_j$. We then have $\lc(\w) = 1$, $\lpp(\w) = x^\alpha\e_j$ and $h = f$. So $f^{[\w]} \in I$ is the desired $f^{[\u']}$.

Next, we deal with the case $g\not=0$, and now $c=\lc(g)/\lc(g_0)$. Let $G := \{g_1^{[\v_1]}, \cdots, g_s^{[\v_s]}\}$ be a full-labeled \gr basis for $I$ such that $\lpp(\v_i) = \t_i$ for $i=1, \cdots, s$.

Since $g^{(x^\alpha\e_j)}$ is an incomplete standard form of $x^\alpha\e_j$ computed from $f^{(x^\alpha\e_j)}$, according to the detailed procedure of Function {\em IncompleteStandardForm}, there exist polynomials $p_1, \cdots, p_s \in R$ such that $$f = p_1 g_1 + \cdots + p_s g_s + g,$$ where $x^\alpha\e_j \succ \lpp(p_i\v_i)$ for $i = 1, \cdots, s$. Note that these $p_i$'s can be obtained similarly as in Function {\em Representation}, but we do not really need to compute them here. As there exists $\fu \in I$ with $\lpp(\u) = x^\alpha\e_j$, let $$\hw := \fu  - p_1(g_1^{[\v_1]}) - \cdots - p_s(g_s^{[\v_s]}) \in I.$$ Clearly, we have $h = g$. Since $\w = \u  - p_1 \v_1 - \cdots - p_s \v_s$ and $\lpp(\u) = x^\alpha\e_j \succ \lpp(p_i\v_i)$, then we have $\lc(\w) = \lc(\u)$ and $\lpp(\w) = \lpp(\u) = x^\alpha\e_j$. Moreover, $g^{[\w]} \in I$ is a standard form of $x^\alpha\e_j$.

Similarly, since $g_0^{(x^\alpha\e_j)}$ is an incomplete standard form of $x^\alpha\e_j$ computed from $(x^\alpha f_j)^{(x^\alpha\e_j)}$, by the detailed procedure of Function {\em IncompleteStandardForm}, there exist polynomials $q_1, \cdots, q_s \in R$ such that $$x^\alpha f_j = q_1 g_1 + \cdots + q_s g_s + g_0,$$ where $x^\alpha\e_j \succ \lpp(q_i\v_i)$ for $i = 1, \cdots, s$. Note that $x^\alpha f_j = (x^\alpha\e_j) \cdot (f_1, \cdots, f_m)$, i.e. $(x^\alpha f_j)^{[x^\alpha\e_j]} \in I$, let $$h_0^{[\w_0]} := (x^\alpha f_j)^{[x^\alpha\e_j]} - q_1 (g_1^{[\v_1]}) - \cdots - q_s(g_s^{[\v_s]}) \in I.$$ Clearly, we have $h_0 = g_0$. Since $\w_0 = x^\alpha\e_j - q_1 \v_1 - \cdots - q_s \v_s$ and $x^\alpha\e_j \succ \lpp(p_i\v_i)$, then we have $\lc(\w_0) = 1$ and $\lpp(\w_0) = x^\alpha\e_j$. And $g_0^{[\w_0]} \in I$ is also a standard form of $x^\alpha\e_j$.

Since $g^{[\w]}$ and $g_0^{[\w_0]}$ are both standard forms of $x^\alpha\e_j$, Proposition \ref{prop_standardform} shows $$\lm(g)/\lc(\w) = \lm(g_0)/\lc(\w_0).$$ Note that $\lc(\w) = \lc(\u)$ and $\lc(\w_0) = 1$, so we obtain $\lc(\u) = \lc(g)/\lc(g_0) = c$. Then $\fu $ is the desired $f^{[\u']}$, which proves the theorem.
\end{proof}

With the above theorem, given a signature-labeled \gr basis $S$ for $I$, we can compute $c_i$ for each $g_i^{(\t_i)} \in S$ such that there exists $g_i^{[\v_i]} \in I$ with $\lm(\v_i) = c_i \t_i$. With these $c_i$'s, then we can construct a monomial-labeled \gr basis for $I$ from $S$.

\smallskip
\noindent {\bf Algorithm --- Sig2MonoLGB}\\
\Input{$S = \{g_1^{(\t_1)}, \cdots, g_s^{(\t_s)}\}$, a signatue-labeled Gr\"obner basis for $I$.}\\
\Output{$M = \{g_1^{\{c_1\t_1\}}, \cdots, g_s^{\{c_s\t_s\}}\}$, a monomial-labeled \gr basis for $I$.}

\begin{enumerate}

\item Let $i:=1$.

\item Choose $g_i^{(\t_i)}$ from $S$, and assume $\t_i$ has the form $x^\alpha \e_j$.

\item Compute an incomplete standard form $g^{(x^\alpha\e_j)}$ of $x^\alpha\e_j$ from $g_i^{(x^\alpha \e_j)}$ by using Function $\psf(g_i^{(x^\alpha \e_j)}, S)$.

\item Compute another incomplete standard form $g_0^{(x^\alpha\e_j)}$ of $x^\alpha\e_j$ from $(x^\alpha f_j)^{(x^\alpha \e_j)}$ by using Function $\psf((x^\alpha f_j)^{(x^\alpha \e_j)}, S)$.

\item If $g\not=0$, let $c_i:=\lc(g)/\lc(g_0)$; otherwise, let $c_i:=1$.

\item If $i = s$, then return $\{g_1^{\{c_1\t_1\}}, \cdots, g_s^{\{c_s\t_s\}}\}$; otherwise, let $i:=i+1$ and goto step 2.

\end{enumerate}




\section{An Illustrative Example} \label{sec_exa}

We will prove in \ref{sec_f5} that F5 computes a signature-labeled \gr basis. In this section, we use the example from \citep{Fau02} to illustrate (1) how to construct a monomial-labeled \gr basis from a signature-labeled \gr basis, (2) how to build a full-labeled \gr basis from a monomial-labeled \gr basis, (3) how to get representations of polynomials in a \gr basis, and (4) how to solve the detachability problem.

\begin{example}
Let $\f = (f_1, f_2, f_3) \in \Q[x, y, z, t]^3$ where $f_1 = yz^3-x^2t^2$, $f_2 = xz^2-y^2t$ and $f_3 = x^2y-z^2t$. The term order $\prec_1$ on $\Q[x, y, z, t]$ is the Degree Reverse Lex order with $x \succ_1 y \succ_1 z \succ_1 t$ and the term order $\prec_2$ on $\Q[x, y, z, t]^3$ is extended from $\prec_1$ in a position over term fashion, i.e.
$$x^\alpha\e_i \prec_2 x^\beta\e_j    \mbox{  iff  } \left\{\begin{array}{l} i > j, \\ \mbox{ or } \\ i = j \mbox{ and } x^\alpha \prec_1 x^\beta. \end{array}\right.$$

A signature-labeled \gr basis returned by F5 is $S = \{g_i^{(\t_i)} \mid i=1, \cdots, 10\}$, where $$g_1^{(\t_1)} = (yz^3-x^2t^2)^{(\e_1)}, \spc g_2^{(\t_2)} = (xz^2-y^2t)^{(\e_2)}, \spc g_3^{(\t_3)} = (x^2y-z^2t)^{(\e_3)}, $$ $$ g_4^{(\t_4)} = (xy^3t-z^4t)^{(xy\e_2)}, \spc g_5^{(\t_5)} = (z^6t-y^5t^2)^{(xyz^2\e_2)},\spc g_6^{(\t_6)} = (y^3zt-x^3t^2)^{(x\e_1)}, $$ $$g_7^{(\t_7)} = (z^5t-x^4t^2)^{(x^2\e_1)}, \spc g_8^{(\t_8)}= (y^5t^2-x^4zt^2)^{(x^2z\e_1)}, $$ $$g_9^{(\t_9)} = (x^5t^2-z^2t^5)^{(x^3\e_1)}, \mbox{ and } \ g_{10}^{(\t_{10})} = (y^6t^2-xy^2zt^4)^{(z^3t\e_1)}.$$
\end{example}

In this example, $I$ is the ideal $\langle f_1, f_2, f_3 \rangle \subset \Q[x, y, z, t]$.

{\bf (1): Construct a monomial-labeled \gr basis $M$ from $S$.}

To construct a monomial-labeled \gr basis from $S$, we need to find the corresponding coefficients for all $\t_i$'s. Here we only take $g_4^{(\t_4)} = (xy^3t-z^4t)^{(xy\e_2)}$ for example, and the other cases are similar.

By Algorithm Sign2MonoLGB, we need to compute two {\em incomplete} standard forms from $(xy^3t-z^4t)^{(xy\e_2)}$ and $(xyf_2)^{(xy\e_2)}$ respectively. Using Function {\em IncompleteStandardForm}, the incomplete standard form of $xy\e_2$ computed from $(xy^3t-z^4t)^{(xy\e_2)}$ is $(xy^3t-z^4t)^{(xy\e_2)}$ itself. Next, we compute an incomplete standard form of $xy\e_2$ from $(xyf_2)^{(xy\e_2)}$ in detail. In Function {\em IncompleteStandardForm}, initially we have $g := xyf_2 = x^2yz^2-xy^3t$. Note that there exists $g_3^{(\t_3)} \in S$ such that $\lpp(g_3) = x^2y$ divides $\lpp(g) = x^2yz^2$ and $(x^2yz^2/x^2y)\lpp(t_3) = z^2\e_3 \prec xy\e_2$. Then let $g:= g - (x^2yz^2/x^2y)g_3 = -xy^3t+z^4t$. For this $g$, there does not exist $g_i^{(\t_i)} \in S$ such that $\lpp(g_i)$ divides $\lpp(g)$ and $(\lpp(g)/\lpp(g_i))\lpp(\t_i) \prec x^\alpha \e_j$. So $(-xy^3t+z^4t)^{(xy\e_2)}$ is the incomplete standard form computed from $(xyf_2)^{(xy\e_2)}$. So  $(xy^3t-z^4t)^{(xy\e_2)}$ and $(-xy^3t+z^4t)^{(xy\e_2)}$ are both {\em incomplete} standard forms of $xy\e_2$, and by Theorem \ref{thm_sbgb2mbgb}, we get $c_4 = -1$.

After obtaining other coefficients $c_i$'s for $g_i^{(\t_i)}$'s, we get a monomial-labeled Gr\"obner basis $$M =\{g_3^{\{\t_3\}}, g_2^{\{\t_2\}}, g_4^{\{-\t_4\}}, g_5^{\{\t_5\}}, g_1^{\{\t_1\}}, g_6^{\{\t_6\}}, g_7^{\{\t_7\}}, g_8^{\{\t_8\}}, g_9^{\{-\t_9\}}, g_{10}^{\{\t_{10}\}}\}.$$ Note that we have sorted the elements in $M$ in an incremental order w.r.t. $\prec_2$ on the $\t_i$'s.

{\bf (2): Build a full-labeled \gr basis $G$ from $M$.}

For this purpose, we use Algorithm Mono2FullLGB. Initially, let $G^{(0)}:=\emptyset$.

{\em Loop 1:} The element with the smallest signature in $M$ is $g_3^{\{\t_3\}} = (x^2y-z^2t)^{\{\e_3\}}$ and then $M := M \setminus \{g_3^{\{\t_3\}}\}$. In Function {\em Representation}, since $h := g_3 - f_3 = 0$, we get a representation $g_3 = f_3$ and all $p_i = 0$. Then $\v_3 := \e_3$, and $G^{(1)} := \{g_3^{[\v_3]}\} = \{(x^2y-z^2t)^{[\e_3]}\} \subset I$.

{\em Loop 2:} The element with the smallest signature in $M$ is $g_2^{\{\t_2\}} = (xz^2-y^2t)^{\{\e_2\}}$ and then $M := M \setminus \{g_2^{\{\t_2\}}\}$. Since $g_2 = f_2$, similarly as Loop 1, we have $\v_2 := \e_2$ and $G^{(2)} := \{g_3^{[\v_3]}, g_2^{[\v_2]}\} =  \{(x^2y-z^2t)^{[\e_3]}, (xz^2-y^2t)^{\{\e_2\}}\} \subset I$.

{\em Loop 3:} The element with the smallest signature in $M$ is $g_4^{\{-\t_4\}} = (xy^3t-z^4t)^{\{-xy\e_2\}}$ and then $M := M \setminus \{g_4^{\{-\t_4\}}\}$. Next, we use Function {\em Representation} to compute a representation of $g_4$. Now $G^{(2)} = \{g_3^{[\v_3]}, g_2^{[\v_2]}\}$. So initially, we have $p_3 := 0$, $p_2 := 0$ and $h := g_4 - (-xy)f_2 = x^2yz^2-z^4t$. On seeing there exists $g_3^{[\v_3]} \in G^{(2)}$ such that $\lpp(g_3) = x^2y$ divides $\lpp(h) = x^2yz^2$ and $(x^2yz^2/x^2y) \lpp(\v_3) = z^2\e_3 \prec xy\e_2$, let $h := h - (x^2yz^2/x^2y) g_3 = 0$ and $p_3 := p_3 + (x^2yz^2/x^2y) = z^2$. Since $h = 0$, we get a representation $$g_4 = -xyf_2+p_3g_3 = -xyf_2 + z^2 g_3.$$ Let $\v_4 := -xy\e_2 + p_3\v_3 = -xy\e_2 + z^2\e_3$, then $(xy^3t-z^4t)^{[-xy\e_2 + z^2\e_3]}$ is a polynomial in $I$, and $G^{(3)} := \{g_3^{[\v_3]}, g_2^{[\v_2]}, g_4^{[\v_4]}\} = \{(x^2y-z^2t)^{[\e_3]}, (xz^2-y^2t)^{\{\e_2\}}, (xy^3t-z^4t)^{[-xy\e_2 + z^2\e_3]}\} \subset I$.

Similarly, we can obtain $\v_4, \cdots, \v_{10}$. At last, we get a full-labeled \gr basis for $I$: $$G = \{g_3^{[\e_3]}, g_2^{[\e_2]}, g_4^{[-xy\e_2 + z^2\e_3]}, g_5^{[xyz^2 \e_2  + y^3t \e_2 - z^4 \e_3]}, g_1^{[\e_1]}, g_6^{[x\e_1 - yz \e_2]}, g_7^{[x^2\e_1 - z^3 \e_3]}, $$ $$ g_8^{[x^2z\e_1 - xyz^2\e_2 - y^3t\e_2]}, g_9^{[-x^3\e_1+yt^2\e_1+z^3t\e_2+xz^3\e_3+t^4\e_3]}, g_{10}^{[z^3t\e_1 - xy^2z^2\e_2 - y^4t\e_2 + xzt^3\e_2 + yz^4\e_3]}\}.$$

{\bf (3): Obtain representations of polynomials in a \gr basis from $G$.}

The set $\{g_3, g_2, g_4, g_5, g_1, g_6, g_7, g_8, g_9, g_{10}\}$ is a \gr basis for $\langle f_1, f_2, f_3\rangle$ by Proposition \ref{prop_gb}, and $G$ provides a representation of each $g_i$ w.r.t. $\{f_1, f_2, f_3\}$ directly. For instance, $g_8^{[x^2z\e_1 - xyz^2\e_2 - y^3t\e_2]} \in G$ indicates $$g_8 = (x^2z\e_1 - xyz^2\e_2 - y^3t\e_2) \cdot (f_1, f_2, f_3) = x^2z f_1 - (xyz^2 + y^3t) f_2.$$

{\bf (4): Solve the detachability problem.}

Let $f := xz^6t - x^5zt^2 + x$ be a polynomial in $\Q[x, y, z, t]$. Reducing $f$ by the set $\{g_3, g_2, g_4, g_5, g_1, g_6, g_7, g_8, g_9, g_{10}\}$, the remainder is $x$ which is not $0$, so $f\notin \langle f_1, f_2, f_3\rangle$.

Let $f := x^6yt^2-xyz^2t^5-xz^6t+x^5zt^2$ be another polynomial in $\Q[x, y, z, t]$. Reducing $f$ by $\{g_3, g_2, g_4, g_5, g_1, g_6, g_7, g_8, g_9, g_{10}\}$, we get:
$$f = xy g_9 - x g_5 - x g_8,$$ which means $f\in \langle f_1, f_2, f_3\rangle$. Next, let $$\u := xy \v_9 - x \v_5 - x \v_8 = (-x^4y + xy^2t^2 - x^3z)\e_1 + xyz^3t \e_2 + (x^2yz^3 + xyt^4 + xz^4) \e_3,$$ which indicates $$f = \u \cdot (f_1, f_2, f_3) = (-x^4y + xy^2t^2 - x^3z) f_1 + xyz^3t f_2 + (x^2yz^3 + xyt^4 + xz^4) f_3.$$

Particularly, note that the set $\{g_3, g_2, g_4, g_1, g_6, g_7, g_8, g_9\}$ is the reduced \gr basis for $I$,  and then representations of polynomials in the reduced \gr basis w.r.t. $\{f_1, f_2, f_3\}$ are as follows:
$$g_3 = f_3, g_2 = f_2, g_4 = -xy f_2 + z^2 f_3, g_1 = f_1, g_6 = x f_1 - yz f_2, g_7 = x^2 f_1 - z^3 f_3, $$ $$g_8 = x^2z f_1 - (xyz^2 + y^3t) f_2,  g_9 = (-x^3 + yt^2) f_1+ z^3t f_2 + (xz^3 + t^4) f_3.$$




\section{Conclusions} \label{sec_con}

A new method to solve the detachability problem of a polynomials is proposed in this paper. The new method only uses the {\em outputs} of signature-based algorithms. To solve the detachability problem, we propose two efficient algorithms. One is to compute full-labeled Groebner basis from a monomial-labeled Groebner basis, the other one is to compute monomial-labeled Groebner basis from a signature-labeled Groebner basis. It is quite easy to check that these two algorithms have polynomial time complexities. Once the full-labeled Groebner basis is known, the detachability problem can be solved directly.

%

\begin{appendix}


\fuluzihao

\section{F5 Computes a signature-labeled \gr Basis} \label{sec_f5}

The proofs in this section are similar to the proofs in \citep{SunWang11}. The proofs are complicated, because these proofs do not depend on the computing order of critical pairs in F5.

Let $\f := (f_1, \cdots, f_m)\in R^m$ and $I$ be the ideal generated by $\{f_1, \cdots, f_m\}$. In F5, the term order $\prec_1$ on $R$ can be any term order, and the term order $\prec_2$ on $R$ is extended from $\prec_1$ in a {\em position over term} fashion. That is, $$x^\alpha\e_i \prec_2 x^\beta\e_j    \mbox{  iff  }
\left\{\begin{array}{l} i > j, \\ \mbox{ or } \\ i = j \mbox{ and
} x^\alpha \prec_1 x^\beta.
\end{array}\right. $$
Thus we have $\e_m \prec_2 \e_{m-1} \prec_2 \cdots \prec_2 \e_1$.

\subsection{F5 Basics} \label{subsec_f5_basic}

Given a term $x^\alpha\e_i$ in $R^m$ and a polynomial $f$ in $R$, we say $f^{(x^\alpha\e_i)}$ is an {\bf admissible labeled polynomial},\footnote{In most papers, such as \citep{Stegers05}, admissible labeled polynomials has the form of $(x^\alpha\e_i, f)$, which is equivalent to the form $f^{(x^\alpha\e_i)}$.} if there exists $\fu \in I$ such that $\lpp(\u) = x^\alpha$. Note that this definition is consistent with our previous definition of $f^{(x^\alpha\e_i)}$. Let $f^{(x^\alpha\e_i)}$ and $g^{(x^\beta\e_j)}$ be two admissible labeled polynomials, $c$ be a constant in $K$ and $t$ be a power product in $R$. Then define: (1) $f^{(x^\alpha\e_i)} + g^{(x^\beta\e_j)} = (f+g)^{(x^\gamma\e_k)}$ where $x^\gamma\e_k = \max_{\prec}\{x^\alpha\e_i, x^\beta\e_j\}$, and (2) $ct (f^{(x^\alpha\e_i)}) = (ctf)^{(ct x^\alpha\e_i)}$. We next introduce some basic definitions in F5.


\begin{define}[Syzygy Criterion]
Let $B$ be a set of admissible labeled polynomials, $f^{(x^\alpha\e_i)} \in B$ be an admissible labeled polynomial, and $t$ be a power product in $R$. We say $t(f^{(x^\alpha\e_i)}) = (tf)^{(tx^\alpha\e_i)}$ is {\bf F5-divisible} by $B$, if there exists $g^{(x^\beta\e_j)}\in B$ with $g\not=0$ such that $\lpp(g)$ divides $tx^\alpha$ and $\e_i \succ \e_j$.
\end{define}

\begin{define}[Rewritten Criterion]
Let $B$ be a set of admissible labeled polynomials, $f^{(x^\alpha\e_i)} \in B$ be an admissible labeled polynomial, and $t$ be a power product in $R$. We say $t(f^{(x^\alpha\e_i)}) = (tf)^{(tx^\alpha\e_i)}$ is {\bf F5-rewritable} by $B$, if there exists $g^{(x^\beta\e_i)} \in B$ such $x^\beta$ divides $tx^\alpha$ and $g^{(x^\beta\e_i)}$ is added to $B$ later than $f^{(x^\alpha\e_i)}$.
\end{define}

Note that the computing order of admissible labeled polynomials in $B$ is very important to  Rewritten Criterion. For convenience, we use an order ``$<$" defined on $B$ to reflect this computing order. Let $f^{(x^\alpha\e_i)}$ and $g^{(x^\beta\e_j)}$ be two admissible labeled polynomials in $B$. We say $g^{(x^\beta\e_j)} < f^{(x^\alpha\e_i)}$, if $g^{(x^\beta\e_j)}$ is added to $B$ {\em later} than $f^{(x^\alpha\e_i)}$. We assume admissible labeled polynomials are added to $B$ one by one, so the order ``$<$" on $B$ is a total order.

\begin{define}[F5-reducible]
Let $f^{(x^\alpha\e_i)}$ be an admissible labeled polynomial and $B$ be a set of admissible labeled polynomials. We say $f^{(x^\alpha\e_i)}$ is {\bf F5-reducible} by $B$, if there exists $g^{(x^\beta\e_j)}\in B$ with $g\not=0$ such that (1) $\lpp(g)$ divides $\lpp(f)$,  denote $t := \lpp(f)/\lpp(g)$, (2) $x^\alpha\e_i \succ tx^\beta\e_j$, and (3) $tg^{(x^\beta\e_j)}$ is neither F5-divisible nor F5-rewritable by $B$.
\end{define}

Given an admissible labeled polynomial $f^{(x^\alpha\e_i)}$ and a set of admissible labeled polynomials $B$, if $f^{(x^\alpha\e_i)}$ is F5-reducible by some $g^{(x^\beta\e_j)}$ in $B$, then we say $f^{(x^\alpha\e_i)}$ one-step-F5-reduces to $(f - c t g)^{(x^\alpha\e_i)}$ by $B$, where $c=\lc(f)/\lc(g)$ and $t=\lpp(f)/\lpp(g)$. If $(f - c t g)^{(x^\alpha\e_i)}$ is still F5-reducible by $B$, we can repeat the above one-step-F5-reduction. We say {\bf $f^{(x^\alpha\e_i)}$ F5-reduces} to $h^{(x^\alpha\e_i)}$ by $B$, if $h^{(x^\alpha\e_i)}$ is obtained by several one-step-F5-reductions from $f^{(x^\alpha\e_i)}$, and $h^{(x^\alpha\e_i)}$ is not F5-reducible by $B$. Based on this reduction procedure, we have the following proposition directly.

\begin{prop} \label{prop_f5reduce}
Let $f^{(x^\alpha\e_i)}$ be an admissible labeled polynomial and $B$ be a set of admissible labeled polynomials. If $f^{(x^\alpha\e_i)}$ F5-reduces to $h^{(x^\alpha\e_i)}$ by $B$, then $h^{(x^\alpha\e_i)}$ is also an admissible labeled polynomial.
\end{prop}

Let $f^{(x^\alpha\e_i)}, g^{(x^\beta\e_j)}$ be two admissible labeled polynomials in $B$ such that $f$ and $g$ are both nonzero. \footnote{We do not care about the critical pairs when either $f$ or $g$ is zero, since these critical pairs make no senses in both practical implementation and theoretical proofs.} Let $t := \lcm(\lpp(f), \lpp(g))$, $t_f := t/\lpp(f)$ and $t_g := t/\lpp(g)$. Then the 4-tuple vector $(t_f, f^{(x^\alpha\e_i)}, t_g, g^{(x^\beta\e_j)})$ is called a {\bf critical pair} of $f^{(x^\alpha\e_i)}$ and $g^{(x^\beta\e_j)}$, if one of the following conditions holds: (1) $t_f x^\alpha\e_i \succ t_g x^\beta\e_j$. And (2) $t_f x^\alpha\e_i = t_g x^\beta\e_j$, and $g^{(x^\beta\e_j)} < f^{(x^\alpha\e_i)}$, i.e. $g^{(x^\beta\e_j)}$ is added to $B$ later than $f^{(x^\alpha\e_i)}$.

Remark that the original F5 does not consider the critical pair of $f^{(x^\alpha\e_i)}$ and $g^{(x^\beta\e_j)}$ if $t_f x^\alpha\e_i = t_g x^\beta\e_j$. We expand the definition of critical pairs here just for theoretical proving. Besides, Proposition \ref{prop_regular} shows that the critical pair $(t_f, f^{(x^\alpha\e_i)}, t_g, g^{(x^\beta\e_j)})$ can be rejected by Rewritten Criterion if $t_f x^\alpha\e_i = t_g x^\beta\e_j$ holds. So the above new definition of critical pairs makes no difference from the original definition in \citep{Fau02}.

For convenience, we say a {\bf critical pair  $(t_f, f^{(x^\alpha\e_i)}, t_g, g^{(x^\beta\e_j)})$ is F5-divisible/F5-rewritable} by $B$ if {either} $t_f(f^{(x^\alpha\e_i)})$ {or} $t_g(g^{(x^\beta\e_j)})$ is F5-divisible/F5-rewritable by $B$. We say $(t_f, f^{(x^\alpha\e_i)}, t_g, g^{(x^\beta\e_j)})$ is a critical pair of $B$, if both $f^{(x^\alpha\e_i)}$ and $g^{(x^\beta\e_j)}$ are in $B$.

\begin{prop} \label{prop_regular}
Let $B$ be a set of admissible labeled polynomials, and $(t_f, f^{(x^\alpha\e_i)}, t_g, g^{(x^\beta\e_j)})$ be a critical pair of $B$. The critical pair $(t_f, f^{(x^\alpha\e_i)}, t_g, g^{(x^\beta\e_j)})$ is F5-rewritable by $B$ if $t_f x^\alpha\e_i = t_g x^\beta\e_j$. Besides, $(f - c t_g g)^{(t_f x^\alpha\e_i)}$ is an admissible labeled polynomial if $t_f x^\alpha\e_i \succ t_g x^\beta\e_j$ where $c=\lc(f)/\lc(g)$.
\end{prop}

\begin{proof}
If $t_f x^\alpha\e_i = t_g x^\beta\e_j$ holds, then $g^{(x^\beta\e_j)}$ is added to $B$ later than $f^{(x^\alpha\e_i)}$ by the definition of critical pairs, and hence, $t_ff^{(x^\alpha\e_i)}$ is F5-rewritable by $g^{(x^\beta\e_j)}\in B$.

Since $f^{(x^\alpha\e_i)}$ and $g^{(x^\beta\e_j)}$ are admissible labeled polynomials, there exist $\fu , \gv \in I$ such that $\lpp(\u) = x^\alpha\e_i$ and $\lpp(\v) = x^\beta\e_j$. Consider the polynomial $t_f(\fu)  - ct_g (\gv)  = (f - c t_g g)^{[t_f\u - c t_g \v]} \in I$ where $c=\lc(f)/\lc(g)$, and clearly, we have $\lpp(t_f\u - c t_g \v) = t_f x^\alpha\e_i$ if $t_f x^\alpha\e_i \succ t_g x^\beta\e_j$, so $(f - c t_g g)^{(t_f x^\alpha\e_i)}$ is an admissible labeled polynomial.
\end{proof}

{\em F5 returns a set of admissible labeled polynomials}, which has been proved in many papers, including \citep{Eder09, Ars09}. So we use this fact and omit detailed proofs in this paper.

The following proposition is very interesting and very important.

\begin{prop} \label{prop_f5_important}
Let $S$ be a finite set of admissible labeled polynomials returned by F5. If $(t_f, f^{(x^\alpha\e_i)}$, $t_g, g^{(x^\beta\e_j)})$ is a critical pair of $S$, then $(t_f, f^{(x^\alpha\e_i)}, t_g, g^{(x^\beta\e_j)})$ is either F5-divisible or F5-rewritable by $S$.
\end{prop}

\begin{proof}
According to F5, the critical pair $(t_f, f^{(x^\alpha\e_i)}, t_g, g^{(x^\beta\e_j)})$ must be considered in some loop of F5, and suppose $S'\subset S$ is the intermediate set when $(t_f, f^{(x^\alpha\e_i)}, t_g, g^{(x^\beta\e_j)})$ is being considered in that loop. Then we have $f^{(x^\alpha\e_i)}, g^{(x^\beta\e_j)}\in S'$.

If $(t_f, f^{(x^\alpha\e_i)}, t_g, g^{(x^\beta\e_j)})$ is either F5-divisible or F5-rewritable by $S'$, then it is also F5-divisible or F5-rewritable by $S$, since $S' \subset S$.

Otherwise, $(t_f, f^{(x^\alpha\e_i)}, t_g, g^{(x^\beta\e_j)})$ is neither F5-divisible nor F5-rewritable by $S'$. In this case, we have $t_f x^\alpha\e_i \succ t_g x^\beta\e_j$ and $(f - c t_g g)^{(t_f x^\alpha\e_i)}$ is an admissible labeled polynomial by Proposition \ref{prop_regular} where $c=\lc(f)/\lc(g)$. Next, according to F5, the labeled polynomial $(f - c t_g g)^{(t_f x^\alpha\e_i)}$ will be F5-reduced by $S'$. Suppose the reduction result is $h^{(t_fx^\alpha\e_i)}$. Proposition \ref{prop_f5reduce} shows $h^{(t_fx^\alpha\e_i)}$ is also an admissible labeled polynomial. Since $h^{(t_fx^\alpha\e_i)}$ is not F5-reducible by $S'$, the labeled polynomial $h^{(t_fx^\alpha\e_i)}$ will be added to $S'$, which means $h^{(t_fx^\alpha\e_i)}\in S$. As $f^{(x^\alpha\e_i)}$ is already in $S'$, then $h^{(t_fx^\alpha\e_i)}$ is added to $S$ later than $f^{(x^\alpha\e_i)}$. So $t_f(f^{(x^\alpha\e_i)})$, and hence $(t_f, f^{(x^\alpha\e_i)}, t_g, g^{(x^\beta\e_j)})$, is F5-rewritable by $h^{(t_fx^\alpha\e_i)}\in S$.
\end{proof}

Combined with Proposition \ref{prop_f5_important}, the following theorem shows F5 computes a signature-labeled Gr\"obner basis.

\begin{theorem} \label{thm_f5}
Let $S$ be a finite set of admissible labeled polynomials. The set $S$ is a signature-labeled Gr\"obner basis for $I$, if both the following two conditions hold.

\begin{enumerate}
\item $\{f_1^{(\e_1)}, \cdots, f_m^{(\e_1)}\}$ is a subset of $S$ and $f_i^{(\e_i)}$ is added to $S$ earlier than any $g^{(x^\beta\e_j)} \in S\setminus \{f_1^{(\e_1)}, \cdots, f_m^{(\e_1)}\}$.

\item For any critical pair $(t_f, f^{(x^\alpha\e_i)}, t_g, g^{(x^\beta\e_j)})$ of $S$, the critical pair $(t_f, f^{(x^\alpha\e_i)}, t_g, g^{(x^\beta\e_j)})$ is either F5-divisible or F5-rewritable by $S$.
\end{enumerate}
\end{theorem}

By Proposition \ref{prop_gb} and the definition of signature-labeled \gr basis, if $S$ is a signature-labeled \gr basis for $I$, then the set $\{f \mid f^{(x^\alpha \e_j)} \in S\}$ is a \gr basis for the ideal $I$, which indicates the correctness of F5.

We will prove Theorem \ref{thm_f5} in the next two subsections. We first rewrite some notations in Subsection \ref{subsec_f5_express}, and then prove an equivalent theorem (Theorem \ref{thm_main}) in Subsection \ref{subsec_f5_proof}.

\subsection{Rewrite Theorem \ref{thm_f5}} \label{subsec_f5_express}

If $f^{(x^\alpha\e_i)}$ is an admissible labeled polynomial, then there exists $\fu \in I$ such that $\lpp(\u) = x^\alpha\e_i$. So we can expand the definitions of F5-divisible, F5-rewritable and critical pairs below. Let $\fu, \gv\in I$ be two polynomials in a set $B$, again, we say $\gv < \fu $ if $\gv $ is added to $B$ {\em later} than $\fu $.

\begin{define}
Let $B$ be a subset of $I$, $\fu$ be a polynomial in $I$ and $t$ be a power product in $R$.
\begin{enumerate}
\item Suppose $\lpp(\u) = x^\alpha \e_i$. We say $t(\fu)$ is {\bf F5-divisible} by $B$, if there exists $\gv \in B$ with $\lpp(\v)=x^\beta \e_j$ and $g\not = 0$, such that $\lpp(g)$ divides $tx^\alpha$ and $\e_i \succ \e_j$.

\item We say $t(\fu) $ is {\bf F5-rewritable} by $B$, if there exists $\gv \in B$ such that $\lpp(\v)$ divides $\lpp(t\u)$, and $\gv  < \fu $ i.e. $\gv $ is added to $B$ later than $\fu $.
\end{enumerate}
\end{define}

Similarly, suppose $\fu, \gv \in B\subset I$ are two polynomials with $f$ and $g$ both nonzero. Let $t:=\lcm(\lpp(f), \lpp(g))$, $t_f:=t/\lpp(f)$ and $t_g:=t/\lpp(g)$. Then the 4-tuple vector $(t_f, \fu, t_g, \gv)$ is called a {\bf critical pair} of $\fu$ and $\gv$, if one of the following conditions holds: (1) $\lpp(t_f\u) \succ \lpp(t_g\v)$. And (2) $\lpp(t_f\u) = \lpp(t_g\v)$, and $\gv < \fu $, i.e. $\gv $ is added to $B$ later than $\fu$. We also denote the critical pair of $\fu$ and $\gv$ by $[\fu, \gv]$ or $[\gv, \fu]$ for short. The corresponding {\bf S-polynomial} of $(t_f, \fu, t_g, \gv)$ is $t_f(\fu) -c t_g (\gv) $ where $c=\lc(f)/\lc(g)$. Similarly, the critical pair $(t_f, \fu, t_g, \gv)$ will be rejected by Rewritten Criterion if $\lpp(t_f\u) = \lpp(t_g\v)$. So the S-polynomial $t_f(\fu) -ct_g(\gv) $ is only considered when $\lpp(t_f\u) \succ \lpp(t_g\v)$. We also say a {\bf critical pair  $(t_f, \fu, t_g, \gv)$ is F5-divisible/F5-rewritable} by $B$ if {either} $t_f(\fu) $ {or} $t_g(\gv) $ is F5-divisible/F5-rewritable by $B$, and $(t_f, \fu, t_g, \gv)$ is a critical pair of $B$ if both $\fu $ and $\gv $ are in $B$.



Then the following theorem is an equivalent version of Theorem \ref{thm_f5}.

\begin{theorem} \label{thm_main}
Let $G$ be a finite subset of $I$. The set $G$ is a full-labeled \gr basis for $I$, if both the following two conditions hold.

\begin{enumerate}
\item $\{f_1^{[\e_1]}, \cdots, f_m^{[\e_1]}\}$ is a subset of $G$ and $f_i^{[\e_i]}$ is added to $G$ earlier than any $\gv  \in G \setminus \{f_1^{[\e_1]}, \cdots, f_m^{[\e_1]}\}$.

\item For any critical pair $[\fu, \gv]$ of $G$, the critical pair $[\fu, \gv]$ is either F5-divisible or F5-rewritable by $G$.
\end{enumerate}
\end{theorem}

\subsection{Proofs of Theorem \ref{thm_main}} \label{subsec_f5_proof}


Let $\fu \in I$, we say $\fu $ has a {\bf standard representation} w.r.t. a set $B\subset I$, if there exist $p_1, \cdots, p_s \in R$ and $g_1^{[\v_1]}, \cdots, g_s^{[\v_s]} \in B$ such that $$f = p_1g_1 + \cdots + p_s g_s,$$ where $\lpp(f)\succeq \lpp(p_ig_i)$ and $\lpp(\u) \succeq \lpp(p_i\v_i)$ for $i=1,\cdots, s$. Clearly, if $\fu $ has a  standard representation w.r.t. $B$, then there exists $\gv \in B$ such that $\lpp(g)$ divides $\lpp(f)$ and $\lpp(\u) \succeq \lpp(t\v)$ where $t=\lpp(f)/\lpp(g)$. We call this fact to be the {\bf basic property of standard representations}.

The following two lemmas are exactly the same as Lemma 3.1 and Lemma 3.2 in \citep{SunWang11} where one can find the detailed proofs.

\begin{lemma} \label{lem_stdrepresentation}
Let $G$ be a finite subset of $I$ and $\{f_1^{[\e_1]},$ $\cdots, f_m^{[\e_m]}\}\subset G$. For a polynomial $\fu \in I$, $\fu$ has a standard representation w.r.t. $G$, if for any critical pair $[\gv, \hw]=(t_g, \gv, t_h, \hw)$ of $G$ with $\lpp(\u) \succeq \lpp(t_g \v)$, the S-polynomial of $[\gv, \hw]$  always has a standard representation w.r.t. $G$.
\end{lemma}

\begin{lemma} \label{lem_correctness}
Let $G$ be a finite subset of $I$ and $\{f_1^{[\e_1]},$ $\cdots, f_m^{[\e_m]}\}\subset G$. Then $G$ is a full-labeled Gr\"obner basis for $I$, if for any critical pair $[\fu, \gv]$ of $G$, the S-polynomial of $[\fu, \gv]$   always has a standard representation w.r.t. $G$.
\end{lemma}

In the proof of Theorem \ref{thm_main}, we need to compare critical pairs, so we introduce the following definitions first. Suppose $(t_f, \fu, t_g, \gv)$ and $({t}_{\brf}, \brfu, {t}_{\brg}, \brgv)$ are two critical pairs of $G$, we say $({t}_{\brf}, \brfu, {t}_{\brg}, \brgv)$ is {\bf smaller} than $(t_f, \fu, t_g, \gv]$ if one of the following conditions holds:
\begin{enumerate}

\item[(a).] $\lpp({t}_{\brf} \bru) \prec \lpp(t_f\u)$.

\item[(b).] $\lpp({t}_{\brf} \bru) = \lpp(t_f\u)$ and $\brfu < \fu$, i.e. $\brfu $ is added to $G$ later than $\fu$.

\item[(c).] $\lpp({t}_{\brf} \bru) = \lpp(t_f\u)$, $\brfu = \fu$ and $\lpp({t}_{\brg} \brv) \prec \lpp(t_g\v)$.

\item[(d).] $\lpp({t}_{\brf} \bru) = \lpp(t_f\u)$, $\brfu = \fu$, $\lpp({t}_{\brg} \brv) = \lpp(t_g\v)$ and $\brgv < \gv$, i.e. $\brgv$ is added to $G$ later than $\gv$.
\end{enumerate}
Let $D$ be a set of critical pairs. A critical pair in $D$ is said to be  {\bf minimal}  if there is no critical pair in $D$  smaller  than this critical pair. Since the order defined on the critical pairs is in fact a total order, the minimal critical pair in $D$ is {\em unique}. We can always find the minimal critical pair in $D$ if $D$ is finite.

Given a critical pair $(t_f, \fu, t_g, \gv)$, there are three possible cases, assuming $c=\lc(f)/\lc(g)$:

\begin{enumerate}

\item If $\lpp(t_f\u - ct_g\v) \not= \lpp(t_f\u)$, then we say $(t_f, \fu, t_g, \gv)$ is {\bf non-regular}.

\item If $\lpp(t_f\u - ct_g\v) = \lpp(t_f\u) = \lpp(t_g\v)$, then $(t_f, \fu, t_g, \gv)$ is called {\bf super regular}.

\item If $\lpp(t_f\u) \succ \lpp(t_g\v)$, then we call $(t_f, \fu, t_g, \gv)$ {\bf genuine regular} or {\bf regular} for short.

\end{enumerate}

Now, we give the proof for Theorem \ref{thm_main}.




\begin{proof}[Proof of Theorem \ref{thm_main}]
We will use Lemma \ref{lem_correctness} to show $G$ is a full-labeled \gr basis for $I$, so we need to consider the critical pairs of $G$. We will take the following strategy. \smallskip \\
{\bf Step 1:}  Let $Todo$ be the set of  {\em all} the critical pairs of $G$, and $Done$ be an empty set.\smallskip \\
{\bf Step 2:} Select the minimal critical pair $[\fu, \gv] = (t_f, \fu, t_g, \gv)$ in $Todo$. \smallskip \\
{\bf Step 3:} For such $[\fu, \gv]$,  we will prove both of the following facts.
\begin{enumerate}
\item[(F1).] The  S-polynomial of  $[\fu, \gv]$ has a standard representation w.r.t. $G$.

\item[(F2).] If $(t_f, \fu, t_g, \gv)$ is {\em super regular} or {\em regular}, then $t_f(\fu) $ is either F5-divisible or F5-rewritable by $G$.
\end{enumerate}
{\bf Step 4:} Move $[\fu, \gv]$ from $Todo$ to $Done$, i.e. $Todo \lla Todo \setminus \{ [\fu, \gv]\}$ and $Done \lla Done\ \cup \{ [\fu, \gv]\}$. \smallskip\smallskip \\
We can repeat {\bf Step 2, 3, 4} until $Todo$ is empty. Please note that for every critical pair in $Done$, it always has property (F1). Particularly, if this critical pair is super regular or regular, then it has both properties (F1) and (F2). When $Todo$ is empty, all the critical pairs of $G$ will lie in $Done$, and hence, all the corresponding S-polynomials  have standard representations w.r.t. $G$. Then $G$ is a full-labeled Gr\"obner basis by Lemma \ref{lem_correctness}.

{\bf Step 1, 2, 4} are trivial, so we next focus on showing the two facts in {\bf Step 3}.

Take the minimal critical pair $[\fu, \gv] = (t_f, \fu, t_g, \gv)$ in $Todo$. The second condition of Theorem \ref{thm_main} shows $[\fu, \gv]$ is either F5-divisible or F5-rewritable by $G$. Then for such $[\fu, \gv]$, it must be in one of the following cases:
\begin{enumerate}

\item[C1:] $(t_f, \fu, t_g, \gv)$ is {\em non-regular}.

\item[C2:] $(t_f, \fu, t_g, \gv)$ is {\em super regular}.

\item[C3:] $(t_f, \fu, t_g, \gv)$ is {\em regular} and $t_f(\fu) $ is {\em either} F5-divisible {\em or} F5-rewritable by $G$.

\item[C4:] $(t_f, \fu, t_g, \gv)$ is {\em regular} and $t_g(\gv) $ is {\em either} F5-divisible {\em or} F5-rewritable by $G$.
\end{enumerate}
Thus, to show the facts in {\bf Step 3}, we have two things to do: First, show (F1) holds in case {C1}; Second, show (F1) and (F2) hold in cases {C2, C3} and {C4}.

We make the following claims under the condition that $(t_f, \fu, t_g, \gv)$ is minimal in $Todo$. The proofs for these claims will be presented after the current proof.
\begin{enumerate}
\item[] {\bf Claim 1:} For any $\brfu \in I$, if $\lpp(\bru) \prec \lpp(t_f\u)$, then $\brfu$ has a standard representation w.r.t. $G$.

\item[] {\bf Claim 2:} If $(t_f, \fu, t_g, \gv)$ is super regular or regular, and $t_f(\fu) $ is F5-divisible or F5-rewritable by $G$, then the S-polynomial of $(t_f, \fu, t_g, \gv)$ has a standard representation w.r.t. $G$.

\item[] {\bf Claim 3:} If $(t_f, \fu, t_g, \gv)$ is regular and $t_g(\gv) $ is either F5-divisible or F5-rewritable by $G$, then $t_f(\fu) $ is either F5-divisible or F5-rewritable by $G$.
\end{enumerate}

Note that Claim 2 indicates that (F2) implies (F1) in the cases {C2, C3} and {C4}, so it suffices to show $t_f(\fu) $ is either F5-divisible or F5-rewritable by $G$ in the cases {C2, C3} and {C4}.

Next, we proceed for each case respectively.

{C1:} $(t_f, \fu, t_g, \gv)$ is {\em non-regular}. Consider the S-polynomial $t_f(\fu) - c t_g (\gv) = (t_f f-c t_g g)^{[t_f\u-ct_g\v]} \in I$ where $c=\lc(f)/\lc(g)$. Note that $\lpp(t_f\u-c t_g\v) \prec \lpp(t_f\u)$ by the definition of non-regular, so Claim 1 shows $(t_f f-c t_g g)^{[t_f\u-ct_g\v]}$ has a standard representation w.r.t. $G$, which proves (F1).

{C2:} $(t_f, \fu, t_g, \gv)$ is {\em super regular}, i.e. $\lpp(t_f\u - ct_g\v) = \lpp(t_f\u) = \lpp(t_g\v)$ where $c=\lc(f)/\lc(g)$. According to the definition of critical pairs, $\gv $ is added to $G$ later than $\fu $. So $t_f(\fu) $ is F5-rewritable by $\gv \in G$.

{C3}: $(t_f, \fu, t_g, \gv)$ is {\em regular} and $t_f(\fu) $ is {\em either} F5-divisible {\em or} F5-rewritable by $G$. (F2) holds naturally.

{C4}: $(t_f, \fu, t_g, \gv)$ is {\em regular} and $t_g(\gv) $ is {\em either} F5-divisible {\em or} F5-rewritable by $G$. {Claim 3} shows $t_f(\fu) $ is  F5-divisible {or} F5-rewritable by $G$ as well.

After all, the theorem is proved.
\end{proof}



We next give the proofs for the three claims appearing in the above proof.

\begin{proof}[Proof of {\bf Claim 1}]
According to the hypothesis, we have $\brfu\in I$ and $\lpp(\bru)\prec \lpp(t_f\u)$. So for any critical pair $(t_{f'}, {f'}^{[\u']}, t_{g'}, {g'}^{[\v']})$ of $G$ with $\lpp(\bru) \succeq \lpp(t_{f'} \u')$, the critical pair $(t_{f'}, {f'}^{[\u']}$, $t_{g'}, {g'}^{[\v']})$ is smaller than $(t_f, \fu, t_g, \gv)$ in fashion (a) and hence lies in $Done$, which means the S-polynomial of $(t_{f'}, {f'}^{[\u']}, t_{g'}, {g'}^{[\v']})$ has a standard representation w.r.t. $G$. Lemma \ref{lem_stdrepresentation} shows that $\brfu$ has a standard representation w.r.t. $G$.
\end{proof}

\begin{proof}[Proof of {\bf Claim 2}]
By hypothesis, we have that $(t_f, \fu, t_g, \gv)$ is minimal in $Todo$, the critical pair $(t_f, \fu, t_g, \gv)$ is super regular or regular, and $t_f(\fu) $ is {either} F5-divisible {or} F5-rewritable by $G$. Let $c := \lc(f)/\lc(g)$. Then $\brfu := t_f (\fu) - c t_g (\gv)  = (t_f f - c t_g g)^{[t_f\u - c t_g \v]} \in I$ is the S-polynomial of $(t_f, \fu, t_g, \gv)$. Since $(t_f, \fu, t_g, \gv)$ is super regular or regular, we have $\lpp(\bru) = \lpp(t_f \u)$. Next we will show that $\brfu$ has a standard representation w.r.t. $G$.

We discuss two cases: (1) $t_f(\fu) $ is F5-divisible by $G$, and (2) $t_f(\fu) $ F5-rewritable and {\em not} F5-divisible by $G$.


(1). If {\em $t_f(\fu) $ is F5-divisible by $G$}, assuming $\lpp(\u) = x^\alpha \e_i$, then by the definition of F5-divisible, there exists $\hw\in G$ with $\lpp(\w)=x^\beta \e_j$ and $h\not = 0$, such that $\lpp(h)$ divides $t_f x^\alpha $ and $\e_i \succ \e_j$. Consider the polynomial $h(f_i^{[\e_i]}) - f_i(\hw) = 0^{[h\e_i - f_i\w]} \in I$, then $\lpp(h\e_i - f_i\w) = \lpp(h)\e_i$ divides $t_f x^\alpha \e_i = \lpp(t_f\u) = \lpp(\bru)$. Let $t_h := (t_f x^\alpha)/\lpp(h)$ and $c := \lc(\u)/\lc(h\e_i - f_i\w)$. For the polynomial $\brfu - c t_h (0^{[h\e_i - f_i\w]}) = \brf^{[\bru - c t_h (h\e_i - f_i\w)]} \in I$, we have $\lpp(\bru - c t_h (h\e_i - f_i\w))\prec \lpp(\bru)$. So $\brf^{[\bru - c t_h (h\e_i - f_i\w)]}$ has a standard representation w.r.t. $G$ by {Claim 1}, which implies $\brfu$ also has a standard representation w.r.t. $G$, since $\lpp(\bru - c t_h (h\e_i - f_i\w))\prec \lpp(\bru)$.



(2). If {\em $t_f(\fu) $ is F5-rewritable and not F5-divisible by $G$}, we need three steps to show $\brfu$ has a standard representation w.r.t. $G$. \smallskip\\
{\bf First:} We show that there exists $f_0^{[\u_0]}\in G$ such that $t_f(\fu) $ is  F5-rewritable by $f_0^{[\u_0]}$ and $t_0(f_0^{[\u_0]})$ is {\em neither} F5-divisible {\em nor}  F5-rewritable by $G$ where $t_0 = \lpp(t_f\u)/\lpp(\u_0)$.\smallskip\\
{\bf Second:} For such $f_0^{[\u_0]}$, we show that $\lpp(\brf) \succeq \lpp(t_0f_0)$ where $t_0 = \lpp(t_f\u)/\lpp(\u_0)$.\smallskip\\
{\bf Third:} We prove that $\brfu$ has a standard representation w.r.t. $G$.\smallskip

Proof of the {\bf First} fact. Since $t_f(\fu) $ is F5-rewritable by $G$, suppose $t_f(\fu) $ is F5-rewritable by some $f_1^{[\u_1]}\in G$, i.e. $\lpp(\u_1)$ divides $\lpp(t_f\u)$ and $f_1^{[\u_1]} < \fu $ which means $f_1^{[\u_1]}$ is added to $G$ later than $\fu $.
Let $t_1 := \lpp(t_f \u) /\lpp(\u_1)$. The polynomial $t_1(f_1^{[\u_1]})$ is not F5-divisible by $G$, since $\lpp(t_f \u) = \lpp(t_1\u_1)$ and $t_f(\fu) $ is not F5-divisible by $G$. If $t_1(f_1^{[\u_1]})$ is not F5-rewritable by $G$, then $f_1^{[\u_1]}$ is the one we are looking for. Otherwise, there exists $f_2^{[\u_2]}\in G$ such that $t_1(f_1^{[\u_1]})$ is F5-rewritable by $f_2^{[\u_2]}$. Note that $t_f(\fu) $ is also F5-rewritable by $f_2^{[\u_2]}$ and we have $\fu  > f_1^{[\u_1]} > f_2^{[\u_2]}$. Let $t_2 := \lpp(t_f \u) /\lpp(\u_2)$. The polynomial $t_2(f_2^{[\u_2]})$ is not F5-divisible by $G$ as well, since $t_f(\fu) $ is not F5-divisible by $G$. We next discuss whether $t_2(f_2^{[\u_2]})$ is not F5-rewritable by $G$. In the better case, $f_2^{[\u_2]}$ is the desired one if $t_2(f_2^{[\u_2]})$ is not F5-rewritable by $G$; while in the worse case, $t_2(f_2^{[\u_2]})$ is F5-rewritable by some $f_3^{[\u_3]} \in G$. We can repeat the above discussions for the worse case. Finally, we will get a chain $\fu  > f_1^{[\u_1]} > f_2^{[\u_2]} > \cdots$. This chain must terminate, since $G$ is finite. Suppose $f_s^{[\u_s]}$ is the last one in the chain. Then $t_f(\fu) $ is F5-rewritable by $f_s^{[\u_s]}$ and $t_s(f_s^{[\u_s]})$ is {\em neither} F5-divisible {\em nor} F5-rewritable by $G$ where $t_s = \lpp(t_f\u)/\lpp(\u_s)$.

Proof of the {\bf Second} fact. From the {\bf First} fact, we have that $t_0(f_0^{[\u_0]})$ is {\em neither} F5-divisible {\em nor} F5-rewritable by $G$ where $t_0 = \lpp(t_f\u)/\lpp(\u_0)$. Next, we prove the {\bf Second} fact by contradiction. Assume $\lpp(\brf) \prec \lpp(t_0f_0)$. Let $c_0 := \lc(\bru)/\lc(\u_0)$. Then for the polynomial $\brfu - c_0 t_0 (f_0^{[\u_0]}) = (\brf - c_0 t_0 f_0)^{[\bru - c_0 t_0 \u_0]} \in I$, we have $\lpp(\brf - c_0 t_0 f_0) = \lpp(t_0 f_0)$ and $\lpp(\bru - c_0 t_0 \u_0) \prec \lpp(\bru) = \lpp(t_0\u_0)$. So $(\brf - c_0 t_0 f_0)^{[\bru - c_0 t_0 \u_0]}$ has a standard representation w.r.t. $G$ by {Claim 1}, and hence, according to the basic property of standard representations, there exists $\hw\in G$ such that $\lpp(h)$ divides $\lpp(\brf - c_0 t_0 f_0) = \lpp(t_0 f_0)$ and $\lpp(t_h\w) \preceq \lpp(\bru - c_0 t_0 \u_0)\prec \lpp(t_0\u_0)$ where $t_h=\lpp(t_0f_0)/\lpp(h)$. Next consider the critical pair $[f_0^{[\u_0]}, \hw]$. Since $\lpp(t_0f_0) = \lpp(t_h h)$, the critical pair $[f_0^{[\u_0]}, \hw]$ has two possible forms.
\smallskip\\
{Form 1:} $[f_0^{[\u_0]}, \hw] = (t_0, f_0^{[\u_0]}, t_h, \hw)$. Since $\lpp(t_0\u_0)\succ \lpp(t_h \w)$, the critical pair $[f_0^{[\u_0]}, \hw]$ is regular and is smaller than $(t_f, \fu, t_g, \gv)$ in fashion (b), which means $[f_0^{[\u_0]}, \hw]$ lies in $Done$ and $t_0(f_0^{[\u_0]})$ is {\em either} F5-divisible {\em or} F5-rewritable by $G$, which contradicts with the property that $t_0(f_0^{[\u_0]})$ is {\em neither} F5-divisible {\em nor} F5-rewritable by $G$.
\smallskip\\
{Form 2:} $[f_0^{[\u_0]}, \hw] = (\bar{t}_0, f_0^{[\u_0]}, \bar{t}_h, \hw)$ where $\bar{t}_0$ divides $t_0$ and $\bar{t}_0 \not= t_0$. Since $\lpp(t_0\u_0) \succ \lpp(t_h \w)$, the critical pair $(\bar{t}_0, f_0^{[\u_0]}, \bar{t}_h, \hw)$ is also regular and is smaller than $(t_f, \fu, t_g, \gv)$ in fashion (a), which means $(\bar{t}_0, f_0^{[\u_0]}, \bar{t}_h, \hw)$ lies in $Done$ and $\bar{t}_0(f_0^{[\u_0]})$ is {\em either} F5-divisible {\em or} F5-rewritable by $G$. Then $t_0(f_0^{[\u_0]})$ is also {\em either} F5-divisible {\em or} F5-rewritable by $G$. This  contradicts with the property that $t_0(f_0^{[\u_0]})$ is {\em neither} F5-divisible {\em nor} F5-rewritable by $G$.
\smallskip
In either case, the {\bf Second} fact is proved.

Proof of the {\bf Third} fact. According to the second fact, we have $\lpp(\brf) \succeq \lpp(t_0f_0)$ where $t_0 = \lpp(t_f\u)/\lpp(\u_0)$. Let $c_0 := \lc(\bru)/\lc(\u_0)$.
For the polynomial $\brfu - c_0 t_0 (f_0^{[\u_0]}) = (\brf - c_0 t_0 f_0)^{[\bru - c_0 t_0 \u_0]} \in I$, we have $\lpp(\brf - c_0 t_0 f _0)\preceq \lpp(\brf)$ and $\lpp(\bru - c_0 t_0 \u_0)\prec \lpp(\bru)$. So $(\brf - c_0 t_0 f_0)^{[\bru - c_0 t_0 \u_0]}$ has a standard representation w.r.t. $G$ by {\bf Claim 1}. Note that $\lpp(\brf)\succeq \lpp(t_0 f _0)$ and $\lpp(\bru)=\lpp(t_0 \u_0)$. So after adding $c_0 t_0 f_0$ to both sides of the standard representation of $\brfu - c_0 t_0 (f_0^{[\u_0]}) = (\brf - c_0 t_0 f_0)^{[\bru - c_0 t_0 \u_0]}$, then we will get a standard representation of $\brfu$ w.r.t. $G$.
\end{proof}

\begin{proof}[Proof of {\bf Claim 3}]
We also prove two cases: (1) $t_g(\gv) $ is F5-divisible by $G$, and (2) $t_g(\gv) $ is F5-rewritable and {\em not} F5-divisible by $G$.

(1). If {\em $t_g(\gv) $ is F5-divisible by $G$}, assuming $\lpp(\v) = x^\alpha \e_i$, then by the definition of F5-divisible, there exists $\hw\in G$ with $\lpp(\w)=x^\beta \e_j$ and $h\not = 0$ such that $\lpp(h)$ divides $t_g x^\alpha $ and $\e_i \succ \e_j$. Consider the polynomial $h(f_i^{[\e_i]}) - f_i(\hw) = 0^{[h\e_i - f_i\w]} \in I$, then $\lpp(h\e_i - f_i\w) = \lpp(h)\e_i$ divides $t_g x^\alpha \e_i = \lpp(t_g\v)$. Let $t_h := (t_g x^\alpha)/\lpp(h)$ and $c := \lc(\v)/\lc(h\e_i)$. For the polynomial $t_g(\gv)  - c t_h (0^{[h\e_i - f_i\w]}) = (t_g g)^{[t_g \v - c t_h (h\e_i - f_i\w)]}  \in I$, we have $\lpp(t_g \v - c t_h (h\e_i - f_i\w))\prec \lpp(t_g\v) \prec \lpp(t_f\u)$. So $(t_g g)^{[t_g \v - c t_h (h\e_i - f_i\w)]}$ has a standard representation w.r.t. $G$ by {\bf Claim 1}. According to the basic property of standard representations, there exists $h'^{[\w']}\in G$ such that $\lpp(h')$ divides $\lpp(t_g g) = \lpp(t_f f)$ and $\lpp(t'_h\w') \preceq \lpp(t_g \v - c t_h (h\e_i - f_i\w)) \prec \lpp(t_g\v)$ where $t'_h=\lpp(t_g g)/\lpp(h')$. Next consider the critical pair $[\fu, h'^{[\w']}]$. Note that $\lpp(t_f f) = \lpp(t_g g) = \lpp(t'_h h')$, so the critical pair of $[\fu, h'^{[\w']}]$ also has two possible forms.
\smallskip\\
{Form 1:} $[\fu, h'^{[\w']}] = (t_f, \fu, t'_h, h'^{[\w']})$. Since $\lpp(t_f\u) \succ \lpp(t_g\v) \succ \lpp(t'_h \w')$, the critical pair $[\fu, h'^{[\w']}]$ is regular and is smaller than $(t_f, \fu, t_g, \gv)$ in fashion (c), which means $[\fu, h'^{[\w']}]$ lies in $Done$ and $t_f(\fu)$ is {\em either} F5-divisible {\em or} F5-rewritable by $G$.\\
{Form 2:} $[\fu, h'^{[\w']}] = (\bar{t}_f, \fu, \bar{t'}_h, h'^{[\w']})$ where $\bar{t}_f$ divides $t_f$ and $\bar{t}_f \not= t_f$. Since $\lpp(t_f\u) \succ \lpp(t_g\v) \succ \lpp(t_h \w)$, the critical pair $(\bar{t}_f, \fu, \bar{t'}_h, h'^{[\w']})$ is also regular and is smaller than $(t_f, \fu, t_g, \gv)$ in fashion (a), which means $(\bar{t}_f, \fu, \bar{t'}_h, h'^{[\w']})$ lies in $Done$ and $\bar{t}_f(\fu)$ is {\em either} F5-divisible {\em or} F5-rewritable by $G$. Then $t_f(\fu)$ is also {\em either} F5-divisible {\em or} F5-rewritable by $G$, since $\bar{t}_f$ divides $t_f$.

(2). {\em The polynomial $t_g(\gv) $ is F5-rewritable and not F5-divisible by $G$}. Since $\lpp(t_g\v) \prec \lpp(t_f\u)$, by using a similar method in the proof of the First and Second facts in {Claim 2}, we have that there exists $g_0^{[\v_0]}\in G$ such that $t_g(\gv) $ is F5-rewritable by $g_0^{[\v_0]}$ and $t_0(g_0^{[\v_0]})$ is {\em neither} F5-divisible {\em nor} F5-rewritable by $G$ where $t_0 = \lpp(t_g\v)/\lpp(\v_0)$. Moreover, we have $\lpp(t_g g) \succeq \lpp(t_0 g_0)$ where $t_0 = \lpp(t_g\v)/\lpp(\v_0)$.

If $\lpp(t_g g) = \lpp(t_0 g_0) = \lpp(t_f f)$, then the critical pair $[\fu, g_0^{[\v_0]}]$ has two possible forms.
\smallskip\\
{Form 1:} $[\fu, g_0^{[\v_0]}] = (t_f, \fu, t_0, g_0^{[\v_0]})$. Since $\lpp(t_f\u) \succ \lpp(t_g\v) = \lpp(t_0 \v_0)$, the critical pair $[\fu, g_0^{[\v_0]}]$ is regular and is smaller than $(t_f, \fu, t_g, \gv)$ in fashion (d), which means $[\fu, g_0^{[\v_0]}]$ lies in $Done$ and $t_f(\fu)$ is {\em either} F5-divisible {\em or} F5-rewritable by $G$.\\
{Form 2:} $[\fu, g_0^{[\v_0]}] = (\bar{t}_f, \fu, \bar{t}_0, g_0^{[\v_0]})$ where $\bar{t}_f$ divides $t_f$ and $\bar{t}_f \not= t_f$. Since $\lpp(t_f\u) \succ \lpp(t_g\v) = \lpp(t_0 \v_0)$, the critical pair $(\bar{t}_f, \fu, \bar{t}_0, g_0^{[\v_0]})$ is also regular and is smaller than $(t_f, \fu, t_g, \gv)$ in fashion (a), which means $(\bar{t}_f, \fu, \bar{t}_0, g_0^{[\v_0]})$ lies in $Done$ and $\bar{t}_f(\fu)$ is {\em either} F5-divisible {\em or} F5-rewritable by $G$. Then $t_f(\fu)$ is also {\em either} F5-divisible {\em or} F5-rewritable by $G$, since $\bar{t}_f$ divides $t_f$.

Otherwise, $\lpp(t_g g) \succ \lpp(t_0 g_0)$ holds. Let $c := \lc(\v)/\lc(\v_0)$. For the polynomial $t_g \gv - c t_0 (g_0^{[\v_0]}) = (t_g g - c t_0 g_0)^{[t_g \v - c t_0 \v_0]}$, we have $\lpp(t_g g - c t_0 g_0) = \lpp(t_g g)$ and $\lpp(t_g \v - c t_0 \v_0) \prec \lpp(t_g \v)$. Then Claim 1 shows $(t_g g - c t_0 g_0)^{[t_g \v - c t_0 \v_0]}$ has a standard representation w.r.t. $G$, and hence, by the basic property of standard representations, there exists $\hw\in G_{end}$ such that $\lpp(h)$ divides $\lpp(t_g g - c t_0 g_0)=\lpp(t_g g)$ and $\lpp(t_h \w) \preceq \lpp(t_g \v - c t_0 \v_0) \prec \lpp(t_g \v)$ where $t_h = \lpp(t_g g)/\lpp(h)$. Note that $\lpp(t_h h) = \lpp(t_g g) = \lpp(t_f f)$. The critical pair of $[\fu, \hw]$ also has two possible forms.
\smallskip\\
{Form 1:} $[\fu, \hw] = (t_f, \fu, t_h, \hw)$. Since $\lpp(t_f\u) \succ \lpp(t_g\v) \succ \lpp(t_h \w)$, the critical pair $[\fu, \hw]$ is regular and is smaller than $(t_f, \fu, t_g, \gv)$ in fashion (c), which means $[\fu, \hw]$ lies in $Done$ and $t_f(\fu)$ is {\em either} F5-divisible {\em or} F5-rewritable by $G$.\\
{Form 2:} $[\fu, \hw] = (\bar{t}_f, \fu, \bar{t}_h, \hw)$ where $\bar{t}_f$ divides $t_f$ and $\bar{t}_f \not= t_f$. Since $\lpp(t_f\u) \succ \lpp(t_g\v) \succ \lpp(t_h \w)$, the critical pair $(\bar{t}_f, \fu, \bar{t}_h, \hw)$ is also regular and is smaller than $(t_f, \fu, t_g, \gv)$ in fashion (a), which means $(\bar{t}_f, \fu, \bar{t}_h, \hw)$ lies in $Done$ and $\bar{t}_f(\fu)$ is {\em either} F5-divisible {\em or} F5-rewritable by $G$. Then $t_f(\fu)$ is also {\em either} F5-divisible {\em or} F5-rewritable by $G$, since $\bar{t}_f$ divides $t_f$.

{\bf Claim 3} is proved.
\end{proof}



\end{appendix}

\end{document}